\documentclass[11pt]{article}

\usepackage{graphicx} 
\usepackage{amsmath}
\usepackage{amssymb} 
\usepackage{bm} 
\usepackage{float}
\usepackage{placeins}
\usepackage{subcaption} 
\usepackage{subcaption}
\usepackage[page]{appendix} 

\usepackage{dcolumn}

\usepackage[margin=1in]{geometry}
\usepackage{setspace}
\allowdisplaybreaks

\usepackage{tikz}

\usepackage{pythonhighlight}

\usepackage{amsthm}

\theoremstyle{definition}
\newtheorem{definition}{Definition}[section]

\newtheorem{lemma}{Lemma}[section]

\newtheorem{theorem}{Theorem}[section]

\newtheorem{proposition}{Proposition}[section]

\newtheorem{corollary}{Corollary}[section]

\usepackage{natbib}
\usepackage{multicol} 
\usepackage{authblk}

\usepackage{xcolor} 

  {\color{purple}} 
  {}                 

\usepackage{mdframed}
\mdfsetup{
  innertopmargin=4pt,
  innerbottommargin=8pt,
  innerleftmargin=8pt,
  innerrightmargin=8pt
}

\usepackage{booktabs}      
\usepackage{float}         
\usepackage{xcolor}        
\usepackage{array}         
\usepackage{caption}       
\usepackage{longtable}     
\usepackage{threeparttable} 
\usepackage{makecell}      
\usepackage{pdflscape}     
\usepackage{graphicx}      
\usepackage{colortbl}      
\usepackage{titlesec}
\usepackage{enumitem}
\usepackage[hidelinks]{hyperref}
\usepackage{xr-hyper}
\usepackage[T1]{fontenc}
\usepackage[utf8]{inputenc}
\usepackage{newtxtext,newtxmath}
\titlespacing*{\section}{0pt}{0.9\baselineskip}{0.45\baselineskip}
\titlespacing*{\subsection}{0pt}{0.75\baselineskip}{0.35\baselineskip}
\titlespacing*{\subsubsection}{0pt}{0.6\baselineskip}{0.3\baselineskip}

\newcommand{\ztil}{\bm{\tilde{z}}}
\newcommand{\xtil}{\bm{\tilde{x}}}
\newcommand{\Atil}{\bm{\tilde{A}}}
\newcommand{\Btil}{\bm{\tilde{B}}}

\newcommand{\pitil}{\bm{\tilde{\pi}}}
\newcommand{\rhotil}{\bm{\tilde{\rho}}}

\begin{document}

\title{\textbf{When Does Static Willingness to Pay Mislead? \\
A Framework for Dynamic Hedonic Valuation}\thanks{\noindent I am especially grateful to Nikhil Agarwal, Isaiah Andrews, Miguel Ballester, Abhijit Banerjee, Nano Barahona, Colin Camerer, Ian Crawford, Glenn Ellison, Drew Fudenberg, Eric Gao, Matt Polisson, Augustus Smith, and Chen Zhen for helpful comments and discussions.}}

\author[ ]{Josephine Auer\thanks{\noindent
E-mail: \texttt{joauer@mit.edu}
}}

\affil[ ]{MIT}

\date{\today}
\maketitle

\begin{abstract}

\noindent Many policy counterfactuals depend on how consumers value product attributes such as sugar, caffeine, alcohol, or emissions. Standard hedonic and differentiated-products models typically impose time-separable preferences. But when attributes are habit forming, current consumption can shift future marginal valuations, so static willingness-to-pay may be insufficient for policy counterfactuals. I develop a nonparametric revealed-preference framework for dynamic hedonic valuation, deriving necessary and sufficient conditions for rationalising observed prices and choices. Using cereal scanner data, I show that the hedonic representation restricts prices, while habits improve behavioural coherence conditional on that representation. The framework diagnoses when static willingness-to-pay is defensible for policy.


\vspace{2ex}



%

\end{abstract}

\newpage
\vspace{-1em}
\section{Introduction}

Taxes, regulations, and product standards often target characteristics rather than products themselves. Soda taxes target sugar, nutrition policy targets sodium and other dietary attributes, tobacco and alcohol policy target addictive substances, and environmental standards target emissions or energy efficiency. Applied economists routinely value such characteristics using hedonic or differentiated-products demand models, and the resulting willingness-to-pay (WTP) measures are then used to interpret welfare and policy counterfactuals.\footnote{Hedonic and characteristic-based demand models are central in empirical industrial organisation and applied micro; see, e.g., \citet{smith_value_1986}, \citet{heckman_importance_1987}, \citet{berry_automobile_1995}, \citet{nevo_measuring_2001}, \citet{gibbons_valuing_2003}, \citet{bajari_estimating_2005}, and \citet{greenstone_does_2008}. For examples valuing potentially habit-forming food attributes under static hedonic preferences, see \citet{dubois_how_2020}, \citet{haeck_estimating_2022}, and \citet{le_fur_willingness_2022}.} The usual interpretation is static: a price or demand coefficient attached to a characteristic is read as a contemporaneous marginal valuation.

This static interpretation is restrictive when the characteristic itself is habit forming. If consuming an additional gram of sugar today lowers tomorrow's marginal utility through fatigue, satiation, or health-related costs, then the observed willingness to pay for that gram combines a current benefit with a negative continuation value. In that case, a static researcher understates contemporaneous WTP by treating the whole price as a current marginal value. If current consumption instead raises future marginal utility through reinforcement or craving, the static interpretation moves in the opposite direction. A static coefficient can sometimes serve as a reduced-form summary of marginal value, but it is not generally a primitive taste parameter and need not be sufficient for counterfactuals that change future preference states.

The central question of this paper is when characteristic-based valuation is coherent, and when it requires a dynamic rather than static interpretation. Existing models of habit formation are typically formulated over \emph{goods}, not characteristics. A large literature studies intertemporal dependence in the consumption of cigarettes, alcohol, food, and digital products by modelling utility as a function of past and current quantities \citep{becker_theory_1988,becker_rational_1991,gruber_tax_2004,demuynck_ill_2013,crawford_habits_2010,allcott_gentzkow_song_2022}. A central lesson of this literature is that price responses are inherently dynamic: past consumption and expected future conditions affect current demand, so temporary and permanent price changes need not be equivalent. For characteristic-based valuation, this raises an additional question. If policy targets an attribute rather than a product, the analyst must know whether \textit{attribute}-level WTP is itself a static object. Existing habit models do not answer this question because they characterise dynamic dependence in product quantities; when valuation is conducted in characteristics space, the dynamic structure must also be specified in characteristics space.

I provide a nonparametric revealed-preference (RP) framework for evaluating the static interpretation of characteristic-based valuation. The framework asks whether observed prices and choices can be rationalised by preferences over current characteristics and lagged consumption of a habit-forming subset of those characteristics. The exercise is consumer by consumer, so valuations over characteristics, discounting, and habit formation may differ freely across households. The payoff is a model-free test of when characteristic WTP has a coherent welfare interpretation, before the researcher imposes a parametric demand specification.

The first contribution is a sharp Afriat-type characterisation of dynamic hedonic rationalisability. The test asks whether there exist characteristic shadow prices and a discount factor such that two requirements hold. First, observed goods prices must admit the shadow-price representation implied by the maintained goods-to-characteristics technology. Second, observed choices must satisfy dynamic RP inequalities at those shadow prices. This extends static characteristics-based rationalisability (e.g., \citet{blow_revealed_2008}) to allow habits over characteristics, while preserving the nonparametric logic of \citet{samuelson_consumption_1948}, \citet{houthakker_revealed_1950}, \citet{afriat_construction_1967}, and \citet{browning_nonparametric_1989}.

The second contribution is to show how dynamics alter the interpretation of characteristic prices. In a static hedonic model, the price of a good purchased in positive quantity equals the value of the characteristics bundled in that good. In the dynamic model, the same price equals contemporaneous marginal valuation plus the discounted continuation value generated by habit-forming attributes. The observed characteristic price can therefore be interpreted as purely contemporaneous marginal value only under an exclusion restriction: current consumption of the characteristic must have no future utility effect. Dynamic wedges also matter for differentiated-products demand, where characteristic coefficients are often converted into WTP for policy-relevant attributes \citep{berry_automobile_1995,nevo_measuring_2001}. Even if such a coefficient is correctly identified in a static demand system, it does not by itself identify how current consumption shifts future preference states. Section~\ref{sec:valuation_wedge} makes this point concrete with a temporary attribute-tax example: static and dynamic models can agree on the immediate price response but disagree about the post-tax consequences of the policy.

Methodologically, the RP exercise plays a different role from stochastic demand estimation. Rather than absorbing deviations into unobserved product characteristics or taste shocks, it characterises the observable implications of the welfare interpretation itself. Latent characteristics can still be added to the model if desired, but then the interpretation being tested is the enlarged one, not the claim that observed policy-relevant attributes alone support static WTP. This diagnostic has substantive content: it makes static WTP as a welfare object empirically refutable, and shows whether any failure points to the price representation, the dynamic specification, or the measurement of the relevant state.

The third contribution operationalises this diagnostic by decomposing empirical failure into two distinct margins. A characteristics-based valuation exercise can fail because observed prices cannot be represented through the maintained characteristics technology; this is a structural failure of the hedonic price system. It can also fail because, conditional on such shadow prices existing, observed choices are not dynamically rationalisable; this is a behavioural failure of intertemporal choice. For applied work, this distinction matters because a failed valuation exercise has different implications depending on whether the problem is the characteristics technology itself or the static interpretation of behaviour at admissible shadow prices. For example, if soda prices do not admit a shadow price for sugar under the maintained characteristics technology, then a soda-tax calculation cannot be grounded in hedonic WTP for sugar. If such shadow prices exist, but only the dynamic model rationalises behaviour, the admissible WTP object combines current utility from sugar with the future utility effects of current sugar consumption.

I complement the binary test with diagnostics that measure how severe violations are on each margin. On the structural margin, I compute the distance from observed prices to the hedonic price manifold implied by the characteristics technology. On the behavioural margin, I adapt a Critical Cost Efficiency Index (CCEI) \citep{afriat1973system,varian_goodness--fit_1990} to measure the smallest relaxation of the dynamic RP inequalities needed for rationalisability. These diagnostics distinguish small departures from consistency from large ones and show whether empirical discipline comes primarily from the price representation or from intertemporal choice.

I apply the framework to household-level scanner data on cold-cereal purchases. Cereal is a natural setting for the exercise because products bundle measurable nutritional and descriptive characteristics, purchases are repeated over time, and certain attributes---such as sugar and sodium---are plausibly habit-forming and relevant for nutrition policy. More generally, the same framework applies to any setting with repeated purchases, observed expenditures, and measured product characteristics.

The application yields three findings. First, moving from goods to characteristics space sharply lowers raw rationalisability rates because characteristic-based models impose real restrictions on the price system. Second, the associated distance measures show that many of these structural violations are modest in dollar terms, so binary rejection alone overstates their economic severity. Third, conditional on the hedonic price representation being admissible, allowing for habits improves behavioural coherence relative to static characteristics models: it raises raw rationalisability rates for some households and moves the distribution of CCEI values closer to one more generally. These patterns show that the framework is not simply a more flexible fit criterion: it diagnoses when static characteristic-based valuation is empirically defensible and when a dynamic interpretation is additionally required.

The results clarify when prices and choices support characteristic-based valuation, when that valuation requires a dynamic interpretation, and what this implies for policy counterfactuals built from static WTP. If observed prices do not admit a characteristics-based representation, no static or dynamic hedonic interpretation is coherent under the maintained technology. If the price representation is admissible but static separability fails, then characteristic prices can no longer be read purely as contemporaneous WTP. In that case, a static WTP coefficient is not generally a primitive taste parameter: it may mix contemporaneous marginal value with omitted habit states and continuation values. The diagnostic asks whether the data support treating static WTP as a defensible welfare object for the counterfactual at hand. The continuation component is generally set identified rather than point identified, so the framework is best understood as a diagnostic of admissibility and interpretation rather than as a full structural estimator of welfare primitives.

The remainder of the paper is organised as follows. Section~2 introduces the dynamic hedonic model, derives the valuation wedge, and uses a temporary-tax example to show how static and dynamic counterfactuals can diverge, making the static interpretation an empirical restriction. Section~3 situates the result within the related RP literature by showing that static hedonic and dynamic goods-based models arise as special cases. Section~4 applies the framework to household scanner data on cereal purchases.

\vspace{-0.35em}
\medskip\noindent\textbf{Related work.} This paper first relates to the hedonic valuation and differentiated-products demand literatures following \citet{rosen_hedonic_1974} and \citet{berry_automobile_1995} (henceforth BLP).\footnote{\citeauthor{rosen_hedonic_1974}'s (\citeyear{rosen_hedonic_1974}) hedonic method has been widely applied across housing, environmental economics, health, education, and industrial organisation; see, for example, \citet{smith_value_1986}, \citet{gibbons_valuing_2003}, \citet{bajari_estimating_2005}, and \citet{greenstone_does_2008}.} Empirical work in this tradition often uses observed prices or demand coefficients to value product characteristics that enter welfare and policy counterfactuals. My point of departure is the static interpretation of those values. When some attributes are habit forming, observed prices or demand coefficients can confound current tastes with continuation values, so the economic interpretation of WTP for characteristics becomes an empirical question rather than a maintained one. This distinction is first-order when WTP estimates are used to evaluate taxes, subsidies, or standards that target the characteristics themselves.

The interpretation problem studied here is distinct from the familiar price-endogeneity problem in differentiated-products demand. In BLP-style models, instruments and supply-side restrictions are used to separate demand from marginal costs, markups, and unobserved product quality. My exercise instead conditions on the prices faced by a household and asks whether those prices and choices admit a characteristic-based welfare interpretation. The two issues are complementary: even if a demand system has recovered a static WTP coefficient for sugar, sodium, or emissions, that coefficient does not reveal whether current consumption of the characteristic shifts future marginal utility. In that sense, the paper adds an interpretation diagnostic to structural hedonic and demand approaches such as \citet{bajari2005demand}: it asks when recovered WTP objects can be read statically, and when counterfactuals require a dynamic state variable.

The paper is also related to the rational addiction literature initiated by \citet{becker_theory_1988}.\footnote{Empirical applications have largely focused on goods such as alcohol, cigarettes, caffeine, and illicit drugs; see, for example, \citet{becker_empirical_1994}, \citet{grossman_empirical_1998}, \citet{gruber_tax_2004}, and \citet{demuynck_ill_2013}.} That literature shows why current consumption may affect future marginal utility, but it is formulated at the level of goods. The distinction matters when persistence is more naturally attached to attributes such as sugar, sodium, nicotine, or caffeine than to the composite products that bundle them \citep[cf.][]{koob_drug_1997}. Relative to empirical work on addiction over nutrient profiles \citep[e.g.][]{richards_native_2006}, my approach replaces functional-form restrictions with a nonparametric RP characterisation and asks when habits over characteristics are consistent with observed prices and choices.

Methodologically, the paper contributes to two nonparametric RP literatures that have developed separately: tests of preferences over characteristics \citep{blow_revealed_2008} and RP analyses of intertemporal dependence \citep{crawford_habits_2010}, both rooted in the foundational work of \citet{afriat_construction_1967}, \citet{diewert_afriat_1973}, and \citet{varian_nonparametric_1982}. I build on these tools to answer a valuation question that neither literature addresses on its own: when do observed prices and choices support a dynamic interpretation of WTP for characteristics? Answering this question requires both a goods-to-characteristics price representation and intertemporal rationalisability at the implied shadow prices. Static hedonic RP tests and dynamic goods-based RP tests therefore arise as limiting cases, but the interaction between the two restrictions is the object of interest. This interaction yields a novel taxonomy of empirical failure and tractable diagnostics for its severity.

Finally, the framework is distinct from models of inventory behaviour and stockpiling in consumer demand \citep{hendel_sales_2006}. Storage can generate persistence in purchases under time-separable preferences, whereas habit formation in my setting operates through state dependence in the flow utility function over characteristics. The empirical framework therefore treats persistence induced by habits as conceptually different from persistence induced by intertemporal substitution, even though distinguishing the two in purchase data requires a maintained link between purchases and the underlying consumption state.

\enlargethispage{1\baselineskip}
Taken together, the paper shows when characteristic-based valuation is admissible, and when the resulting valuations require a dynamic rather than static interpretation. Moving to characteristics space imposes demanding structural restrictions on the admissible price system; allowing for dynamics changes the set of admissible behavioural patterns and the interpretation of shadow prices conditional on satisfying those restrictions.

\section{Model}
\label{sec:model}

I begin by formalising a dynamic hedonic environment in which goods map into characteristics and habits attach to a subset of those characteristics.
I observe a single consumer for $t=1,\ldots,T$, with purchases $\bm{x}_t\in\mathbb{R}^K_+$ and present-value prices $\bm{\rho}_t\in\mathbb{R}^K_+$. Observed goods prices are taken as given from the consumer's perspective, so the object of the analysis is not price determination but whether observed choices can be rationalised. Goods map into $J$ measured characteristics via a time-invariant linear technology $\bm{z}_t=\bm{A}\bm{x}_t$ following \citet{gorman_possible_1956}, where $\bm{A}$ is a $J\times K$ matrix (typically $J<K$).\footnote{I adopt a linear transformation from goods to characteristics space as it is the most widely used specification. Most results extend to a non-linear setting where $\bm{z}=\bm{F}(\bm{x})$ is increasing and strictly concave; see Supplemental Appendix~\ref{sec:nonlinear_char} for an analogue of the consistency definition. The key difference to the linear case arises in the marginal product: whereas $\partial \bm{z}/\partial \bm{x}=\bm{A}^{\prime}$ is constant in the linear model, the marginal product varies with demand under a non-linear transformation.} I treat $\bm{A}$ as known and stable over time; it captures objectively defined product attributes, while valuation is encoded in preferences. Although the model is written for one consumer, the empirical RP exercise is applied household by household, so valuations over characteristics, discounting, and habit formation are all allowed to vary freely across consumers. 

I partition characteristics as $\bm{z}_t=((\bm{z}_t^c)^{\prime},(\bm{z}_t^a)^{\prime})^{\prime}$, where $\bm{z}_t^c\in\mathbb{R}^{J_1}$ are non-habit-forming characteristics and $\bm{z}_t^a\in\mathbb{R}^{J_2}$ are habit-forming characteristics, with $J_1+J_2=J$. The analyst defines this partition. Crucially, this formulation allows intertemporal dependence to operate at the level of attributes rather than goods, so persistence in behaviour need not be attributed to non-habit-forming components bundled within a product. Although preferences are defined over characteristics, choice and budget constraints remain in goods space.

To clarify the economic content of the model, it is useful to fix ideas with a simple example. Suppose the consumer chooses between two cereal products, where each good bundles two measurable attributes: a contemporaneous ``taste'' characteristic (e.g., nuttiness) and a habit-forming ``sensory'' characteristic (e.g., salt or sugar intensity). The key modelling choice is that habits attach to the latter attribute rather than to the cereal good itself: consuming a high-intensity product today can change tomorrow's marginal value of intensity (due to sensory fatigue or craving), even if the consumer switches cereal product. In this example, the model's structural content is that observed goods prices must be representable as shadow values on attributes given $\bm{z}_t=\bm{A}\bm{x}_t$, while its behavioural content is that those shadow values must admit a concave, dynamically consistent utility representation. 

Preferences are represented by a felicity (i.e., flow utility) function $u:\mathbb{R}^{J+J_2}\to\mathbb{R}$ that depends on current characteristics and one lag of the habit-forming subset, $u(\bm{z}_t^c,\bm{z}_t^a,\bm{z}_{t-1}^a)$, as in the one-lag ``short memory habits'' specification of \citet{boyer_habit_1978,boyer_rational_1983} and \citet{becker_empirical_1994}. The multi-lag extension is straightforward and deferred to Supplemental Appendix~\ref{sec:extended_lag}. I assume quasi-linearity in an outside good $y_t$ with unit price, as is standard in empirical IO and hedonic demand models for narrow product categories \citep{berry_automobile_1995,nevo_measuring_2001}. I also maintain local non-satiation, concavity, and superdifferentiability of $u$; I impose no further sign or monotonicity restrictions on the habit-forming components. 
I refer to this environment as the \emph{habits-over-characteristics} model.

The consumer chooses $\{(\bm{x}_t,y_t)\}_{t=1}^T$ to solve
\begin{equation}
\label{max_prob}
\max_{\{(\bm{x}_t,y_t)\}_{t=1}^{T}}\sum_{t=1}^{T}\beta^{t-1}\big(u(\ztil_t)+y_t\big)\quad\text{s.t.}\quad\sum_{t=1}^{T}\bm{\rho}_t^{\prime}\bm{x}_t+\sum_{t=1}^{T}\beta^{t-1}y_t=W,\qquad \ztil_t=\Atil\,\xtil_t,
\end{equation}
where $\beta\in(0,1]$ is a discount factor and $W$ is present-value lifetime wealth. The single present-value budget should be read as the standard quasilinear lifecycle formulation: the consumer can reallocate expenditure across observed periods, while the outside good absorbs wealth effects. I use the augmented notation
\begin{equation}
\label{augmented_notation}
\ztil_t:=\big((\bm{z}_t^c)^{\prime},(\bm{z}_t^a)^{\prime},(\bm{z}_{t-1}^a)^{\prime}\big)^{\prime},\quad \xtil_t:=(\bm{x}_t^{\prime},\bm{x}_{t-1}^{\prime})^{\prime},\quad \Atil:=\begin{pmatrix}\bm{A}&\bm{0}_{J\times K}\\ \bm{0}_{J_2\times K}&\bm{A}^a\end{pmatrix},
\end{equation}
so $\ztil_t\in\mathbb{R}^{J+J_2}$, $\xtil_t\in\mathbb{R}^{2K}$, and $\Atil$ is a $(J+J_2)\times 2K$ block matrix. Throughout, I treat $\bm{x}_0$ (equivalently, the initial habit stock $\bm{z}_0^a$) as fixed and exogenous. Because the model has one lag and a finite horizon, there is no continuation term beyond $T$; equivalently, I set $\bm{\pi}_{T+1}^1 \equiv \bm{0}$ by convention. The key question is whether the observables $\{(\bm{\rho}_t,\bm{x}_t)\}_{t=1}^T$ can be rationalised by \eqref{max_prob}, and, if so, what testable restrictions this imposes on prices and choices.

\vspace{-1em}
\subsection{Consistency}

I now ask whether the observed data can be rationalised by optimising behaviour under the habits-over-characteristics model.
This notion of consistency encompasses both the existence of a hedonic shadow-price representation (a structural requirement) and the coherence of intertemporal behaviour conditional on that representation.
Here ``structural'' does not refer to price determination. Rather, it refers to whether the maintained goods-to-characteristics technology can rationalise observed prices through a corresponding system of characteristic shadow prices.

\begin{definition}[\textbf{Consistency}]
\label{consistency_definition}
The data $\{(\bm{\rho}_t,\bm{x}_t)\}_{t=1}^T$ are \emph{consistent} with the one-lag habits-over-characteristics model for given $\bm{A}$ if there exist $\beta\in(0,1]$ and a locally non-satiated, concave, superdifferentiable felicity function $u$ such that $\{\bm{x}_t\}_{t=1}^T$ solves \eqref{max_prob}.
\end{definition}

The following lemma unpacks Definition~\ref{consistency_definition} into the first-order pricing restrictions implied by the consumer's problem. The main revealed-preference characterisation, which eliminates the unknown utility function, is stated in Theorem~\ref{Afriat}.

\vspace{2ex}
\begin{mdframed}
\begin{lemma}[\textbf{First-Order Pricing Conditions}]
\label{consistency_conditions}
The data $\{\bm{\rho}_t,\bm{x}_t\}_{t=1}^T$ are consistent with the one-lag habits-over-characteristics model for a given technology matrix $\bm{A}$ if and only if there exist a locally non-satiated, concave, superdifferentiable utility function $u(\cdot)$ and a discount factor $\beta\in(0,1]$ such that for all $t\in\{1,\ldots,T\}$,
\begin{align}
\label{foc_market_prices}
\bm{\rho}_t 
\geq    
\bm{A}^{\prime} \bm{\pi}_t^{0}
+
(\bm{A}^a)^{\prime} \bm{\pi}_{t+1}^{1},
\tag{$\star$}
\end{align}
with equality for all goods $k$ such that $x_t^k>0$, where the discounted shadow prices are
\begin{align}
\label{shadow_price_0}
\bm{\pi}_t^{0} &= 
\beta^{t-1}
\begin{bmatrix}
\partial_{\bm{z}_t^c} u(\ztil_t) \\
\partial_{\bm{z}_t^a} u(\ztil_t)
\end{bmatrix}, \\
\label{shadow_price_1}
\bm{\pi}_{t}^{1} &=
\beta^{t-1}
\partial_{\bm{z}_{t-1}^a} u(\ztil_t),
\end{align}
and $\ztil_t=\Atil\xtil_t$, with $\bm{\pi}_{T+1}^1 \equiv \bm{0}$.
\end{lemma}
\end{mdframed}

\textit{Proof:} See Appendix~\ref{sec:proofs}. $\qed$

\vspace{-1em}
\subsection{Dynamic wedges in willingness to pay}
\label{sec:valuation_wedge}

Lemma~\ref{consistency_conditions} shows why habits can matter for characteristic-based valuation. The shadow price $\bm{\pi}_t^0$ is the discounted marginal valuation of contemporaneous characteristics, while $\bm{\pi}_t^{1}$ captures the marginal utility effect of past consumption of habit-forming characteristics. For any good purchased in positive quantity, condition $(\star)$ implies that its price equals the sum of current characteristic values and the discounted continuation value generated by habit-forming attributes.

This decomposition creates a wedge between static and dynamic WTP. In a static hedonic model, the price of a good is interpreted as the value of the characteristics bundled in that good. With habits, the same price also internalises the future utility effect of current consumption. For example, when current consumption lowers future marginal utility (i.e., $\partial_{\bm{z}_t^a} u(\ztil_{t+1}) < 0$), goods prices satisfy
\[
\rho_t^k
=
\bm{a}_k^{\prime}\bm{\pi}_t^0
+
\bm{a}_k^{a\prime}\bm{\pi}_{t+1}^{1}
<
\bm{a}_k^{\prime}\bm{\pi}_t^0,
\]
so ignoring intertemporal dependence understates contemporaneous WTP for goods with negatively reinforcing attributes. If current consumption instead raises future marginal utility through reinforcement or craving, the inequality reverses. The observed characteristic price can therefore be interpreted as purely contemporaneous marginal value only when the continuation component is zero.

The distinction matters for counterfactuals that change current consumption of habit-forming characteristics. To see this, consider the one-characteristic case with unit characteristic content, so that $a_k=a_k^a=1$. For an interior purchase, the pricing condition reduces to
\[
\rho_t^k=\pi_t^0+\pi_{t+1}^1,
\]
where $\pi_t^0$ is the discounted contemporaneous marginal value of the characteristic and $\pi_{t+1}^1$ is the discounted marginal effect of today's characteristic consumption on tomorrow's utility. A static model observes the sum $\rho_t^k$ and treats it as a current marginal valuation. The dynamic shadow price may still be the consumer's relevant private lifetime marginal value at the observed state. What the static representation lacks is the state variable that carries current characteristic consumption forward: in the dynamic model, current $\bm{z}_t^a$ enters tomorrow's utility.

\vspace{-1em}
\paragraph{Example: a temporary attribute tax.} Suppose current sugar consumption raises future marginal utility for sugar through reinforcement. A temporary sugar tax today reduces current sugar consumption and, by lowering tomorrow's habit stock, also shifts tomorrow's demand for sugar after the tax has expired. The intuition is simple: if consuming sugar today makes sugar more valuable tomorrow, then reducing sugar today lowers tomorrow's desire for sugar as well. A static model can capture the contemporaneous price response, but it cannot describe how the induced change in $\bm{z}_t^a$ changes future marginal valuations. The same nominal tax imposed later need not have the same effect, because the consumer reaches the later date with a different habit stock.

A two-period calculation makes the policy content transparent. Continue with the one-good, one-characteristic case above, so that $A=A^a=1$ and $z_t=x_t$: buying one unit of the good delivers one unit of the habit-forming characteristic. Assume wealth is high enough that the interior solution is feasible; quasi-linearity then removes income effects from the choice of $z_1$ and $z_2$. Finally, assume no initial habit stock, $\beta=1$, and concave felicity
\[
u(z_t,z_{t-1})=3z_t-\frac{1}{2}z_t^2+\frac{1}{2}z_tz_{t-1}-\frac{3}{8}z_{t-1}^2.
\]
The cross-partial $\partial^2 u/\partial z_t\partial z_{t-1}=1/2$ captures reinforcement: a higher lagged habit stock raises the current marginal value of the characteristic.
With no tax and prices $\rho_1=\rho_2=1$, the interior first-order conditions imply
\[
z_2=2+\frac{1}{2}z_1,\qquad \frac{7}{4}z_1=2+\frac{1}{2}z_2,
\]
so $z_1=2$ and $z_2=3$. Now impose a temporary tax $\tau=0.3$ only in period 1. The first-order conditions become
\[
z_2=2+\frac{1}{2}z_1,\qquad \frac{7}{4}z_1=1.7+\frac{1}{2}z_2,
\]
so $z_1=1.8$ and $z_2=2.9$. Even if a time-separable static model matched the period-1 reduction from $2$ to $1.8$, it would predict no post-tax effect in period 2 once the tax expires, because period-2 utility contains no state inherited from period 1. The dynamic model instead predicts an additional reduction from $3$ to $2.9$ because the period-1 tax changes the habit stock entering period-2 utility.

The same calculation also illustrates the price wedge in \eqref{foc_market_prices}. In period 1, the tax-inclusive price equals contemporaneous marginal value plus the continuation value of current consumption:
\[
\rho_1+\tau
=
\underbrace{(3-z_1)}_{\pi_1^0}
+
\underbrace{\left(\frac{1}{2}z_2-\frac{3}{4}z_1\right)}_{\pi_2^1}.
\]
At the taxed allocation, this gives $1.3=1.2+0.1$. A static interpretation can recover the full net marginal value $1.3$ at the observed state. What it cannot recover is the state-dependence in that value. In the calculation, the continuation component is $\pi_2^1=0.1$: it is the marginal value of changing the lagged characteristic that enters period-2 utility. The same object explains the post-tax demand effect, since the period-2 first-order condition gives $z_2=2+\frac{1}{2}z_1$. When the period-1 tax lowers $z_1$, it lowers period-2 demand even though the tax has expired.

The external-cost consequences of the policy can therefore change. Suppose sugar consumption generates an external social cost $h z_t$ that is not internalised by the consumer. A static evaluation of the temporary tax counts only the period-1 external-cost saving, $h(2-1.8)=0.2h$. The dynamic evaluation counts the same contemporaneous saving plus the post-tax saving, $h(3-2.9)=0.1h$, for a total external-cost saving of $0.3h$. Thus even when the two models agree on the immediate effect of the tax, they disagree about its total external-damage consequences. More generally, any policy calculation that trades these benefits against implementation costs or other welfare components can deliver a different ranking. Static WTP may therefore be adequate for valuing a marginal unit at the observed state, but it is not generally sufficient for evaluating policies whose effects propagate through habit stocks. The framework below provides a diagnostic for whether the data support a static interpretation, or instead point to dynamic state effects that a counterfactual should take seriously.

The same logic applies to differentiated-products demand models \citep{berry_automobile_1995,nevo_measuring_2001}. These models are often used to evaluate taxes, subsidies, or product standards that move consumption of characteristics such as sugar, sodium, emissions, or energy efficiency. If preferences are dynamically dependent over those characteristics, a static coefficient may still recover net WTP at the observed state, but it need not be sufficient for counterfactuals that move current consumption of habit-forming characteristics. This restriction is distinct from the usual equilibrium price problem in BLP-style demand estimation, where instruments and supply-side restrictions separate demand from marginal costs, markups, and unobserved product quality. Even if standard instruments and supply-side restrictions successfully identify the relevant static demand coefficient, they do not by themselves identify how current consumption shifts future preference states.

For an applied demand researcher, the practical role of the test is as a diagnostic pre-test and interpretation check. Given a maintained characteristics technology, the first step is to ask whether observed prices support characteristic shadow values. Conditional on that representation, the second step is to ask whether behaviour at those shadow values is consistent with static separability, or whether allowing current characteristics to affect future marginal utility improves rationalisability. The distance-based diagnostics developed below are useful even when the model rejects: large violations warn against treating recovered characteristic WTP as a static welfare primitive, while small violations support the static model as a useful approximation. The framework is therefore complementary to structural demand estimation: it clarifies when the WTP objects recovered by such models can be read statically, and when they require a dynamic interpretation.

In sum, the model’s empirical content is governed by two restrictions: a \emph{structural} requirement that observed goods prices admit characteristic shadow prices consistent with the maintained goods-to-characteristics technology, and a \emph{behavioural} requirement that these shadow prices be consistent with concave, dynamically stable preferences.

\vspace{-1em}
\subsection{Afriat conditions for habits-over-characteristics} 
\label{Afriat_thm}

I now state the central theoretical result of the paper. It shows that the model has two distinct sources of empirical content: a structural requirement that observed goods prices admit characteristic shadow prices given $\bm{A}$, and a behavioural requirement that these shadow prices be consistent with concave, dynamically stable preferences.

\vspace{1em}
\begin{mdframed}
\begin{theorem}[]
\label{Afriat}
The following statements are equivalent:
\begin{itemize}
\item[(A)] 
The data $\left\{\bm{\rho}_t; \bm{x}_t\right\}_{t \in \{1, \ldots, T\}}$ are consistent with the one-lag habits model for given technology $\bm{A}$.

\item[(B)]
There exist 
$T$ $J$-vector discounted shadow prices
$\left\{\bm{\pi}_t^{0}\right\}_{t \in\{1, \ldots, T\}}$, 
$T$ $J_2$-vector discounted shadow prices
$\left\{\bm{\pi}_t^{1}\right\}_{t \in\{1, \ldots, T\}}$ and a discount factor $\beta\in(0,1]$ such that,
\begin{align}
& \label{(B1)}
0 \leq \sum_{m=1}^{M} \pitil_{t_m}^{\prime}\left(\ztil_{t_{m+1}}-\ztil_{t_m}\right)
&& \forall \, M \geq 2,\ \forall \, (t_1,\ldots,t_M)\in\{1,\ldots,T\}^{M},\ t_{M+1}=t_1
 \tag{B1} \\
& \label{(B2)}
\rho^k_t
\geq 
\bm{a}_k^{\prime} \bm{\pi}_t^0 +
\bm{a}_k^{a \prime}  \bm{\pi}_{t+1}^{1} 
&& \forall \, k, t \in\{1, \ldots, T\} 
 \tag{B2} \\
 & \label{(B3)}
\rho^k_t
= 
\bm{a}_k^{\prime} \bm{\pi}_t^0 +
\bm{a}_k^{a\prime}  \bm{\pi}_{t+1}^{1} 
&& \text{if } x_t^k > 0, \, \, \forall \, k, t \in\{1, \ldots, T\}
 \tag{B3}
\end{align}
where 
$\bm{a}_k$ is the $J$-vector corresponding to the $k$-th column of $\bm{A}$,
$\bm{a}_k^a$ is the $J_2$-vector corresponding to the last $J_2$ rows of the $k$-th column of $\bm{A}$,
and 
$\pitil_t := \frac{1}{\beta^{t-1}}\left[\bm{\pi}_t^{0 \prime}, \bm{\pi}_t^{1 \prime}\right]^{\prime}$, with $\bm{\pi}_{T+1}^1 \equiv \bm{0}$.

\end{itemize}
\end{theorem} 
\end{mdframed}

\textit{Proof:} See Appendix~\ref{sec:proofs}.  $\qed$

\medskip
\label{Afriat_discussion}

Theorem \ref{Afriat} says exactly when the observed purchases can be interpreted as optimal choices over characteristics with one-period habits. Formally, the data are rationalisable if and only if there exist shadow prices and a discount factor satisfying conditions \eqref{(B1)}--\eqref{(B3)}. When such objects exist, one can construct a concave, locally non-satiated utility function over characteristics that rationalises observed choices; when they do not, no such representation is possible. Given the boundary convention stated above, the final-period pricing restriction is simply \eqref{(B2)}--\eqref{(B3)} with $\bm{\pi}_{T+1}^1=\bm{0}$. The theorem therefore delivers a sharp, nonparametric test of dynamic consistency in characteristics space.

It is useful to interpret the economic content of the three conditions. Condition \eqref{(B1)} imposes cyclical monotonicity on the shadow prices. Economically, it rules out ``cycles'' in revealed marginal valuations over augmented characteristic bundles: there should be no sequence of observed trades in characteristics space that would allow a costless improvement by returning to the starting point. Formally, cyclical monotonicity is equivalent to concavity of the instantaneous utility function \citep{rockafellar_convex_1970}. This condition is precisely the behavioural discipline of the model.

Conditions \eqref{(B2)} and \eqref{(B3)} impose the structural pricing restrictions implied by the habits-over-characteristics model. They require that observed goods prices be representable as linear combinations of contemporaneous and forward-looking shadow prices. Economically, current goods prices must internalise both current marginal utility from characteristics and the continuation value induced by habit formation. In the running cereal example, the price of a high salt or sugar cereal must reflect not only current taste for nuttiness and sensory intensity, but also how today's sensory intensity alters tomorrow's marginal utility. These equalities therefore encode the intertemporal wedge introduced by past consumption directly into the price system. 

Finally, note that the mechanism tested here differs conceptually from inventory-driven persistence \citep{hendel_sales_2006}: stockpiling generates serial correlation in purchases through intertemporal substitution and storage under time-separable preferences. Contrastingly, persistence in my framework operates through state dependence in utility over characteristics, so current consumption of habit-forming attributes carries a continuation value by shifting future marginal valuations. In purchase data, however, the two mechanisms need not be empirically separable without additional structure linking purchases to consumption.

In practice, the test is implemented as a sequence of linear feasibility problems. For a fixed value of $\beta$, the researcher constructs the augmented characteristic bundles $\tilde{\bm{z}}_t$, searches for shadow prices satisfying the pricing restrictions \eqref{(B2)}--\eqref{(B3)}, and checks whether those shadow prices satisfy the Afriat inequalities. Gridding over $\beta$ then gives the reported rationalisability result. The apparently large family of cycle restrictions in \eqref{(B1)} can be replaced by pairwise Afriat inequalities, so the implementable problem contains only a quadratic number of behavioural constraints; Supplemental Appendix~\ref{sec:LP} gives the explicit linear programme.

Interpreting the Afriat test requires care. A positive result establishes existence: there exists \emph{some} concave utility function and discount factor consistent with the data. The representation, however, is not unique. Distinct utility functions---beyond simple monotone transformations---may rationalise the same dataset. Moreover, rationalisability depends on the level of temporal and product aggregation; for example, time aggregation may smooth consumption in a way that mimics habit persistence.

A negative result is likewise not diagnostic of the precise source of failure. Rejection may reflect habit persistence extending beyond one lag, non-concavities in preferences, misspecification of the characteristics technology, or an incorrect partition of $\bm{A}$ into contemporaneous and habit-forming components (for instance, treating ``taste'' preferences like nuttiness as static when the data in fact suggest that exposure today shifts future marginal valuations). The Afriat test should therefore be interpreted as a sharp but reduced-form diagnostic of dynamic rationalisability, rather than as a structural identification device.

\vspace{-1em}
\subsection{A necessary price-span screen}
\label{sec:nec_rank_cond}

Before running the full feasibility test, one can apply a low-dimensional diagnostic that asks whether observed prices can in principle be written as characteristic shadow values. This subsection derives the corresponding necessary condition. It can be evaluated period by period: failure immediately falsifies the model, while satisfaction means only that the data pass this preliminary price screen.

The restriction is \emph{structural} in the sense that it concerns only whether the maintained goods-to-characteristics technology can in principle support \emph{any} shadow-price representation of observed prices. If the condition fails, no candidate preferences---static or dynamic---can rationalise the data because the implied shadow-price system does not exist.

Let $\bm{\rho}_t^+$ denote the $K_t^+\le K$-dimensional vector of discounted prices of goods consumed in strictly positive quantities in period $t$. Let $\bm{B}_t$ denote the $J\times K_t^+$ submatrix of $\bm{A}$ collecting the contemporaneous characteristics of those goods, and let $\bm{B}_t^a$ denote the $J_2\times K_t^+$ submatrix collecting the corresponding habit-forming characteristics. For interior dates, Theorem~\ref{Afriat} implies that
\begin{align}
\bm{\rho}_t^+
&= \bm{B}_t^{\prime}\bm{\pi}_t^0 + (\bm{B}_t^a)^{\prime}\bm{\pi}_{t+1}^1 
= 
\begin{bmatrix}
\bm{B}_t^{\prime} \,\big\rvert\, (\bm{B}_t^a)^{\prime}
\end{bmatrix}
\begin{bmatrix}
\bm{\pi}_t^0\\
\bm{\pi}_{t+1}^1
\end{bmatrix}
=: \Btil_t
\begin{bmatrix}
\bm{\pi}_t^0\\
\bm{\pi}_{t+1}^1
\end{bmatrix},
\label{rank_equation}
\end{align}
where $\Btil_t$ is the $K_t^+\times (J+J_2)$ augmented technology matrix. Because $\bm{B}_t^a$ is formed by rows of $\bm{B}_t$, the augmented matrix $\Btil_t$ has the same column space as $\bm{B}'_t$, so habit formation tilts shadow prices within the same $J$-dimensional price manifold rather than expanding it. At a genuine terminal date, the same pricing equation holds with $\bm{\pi}_{T+1}^1=\bm{0}$; the interior-date formulation is the one directly relevant for the empirical implementation, which does not treat the final observed purchase period as terminal.

Consistency therefore requires the observed price vector $\bm{\rho}_t^+$ to lie in the column space of $\Btil_t$, yielding the following necessary condition.
\vspace{1em}
\begin{mdframed}
\begin{proposition}
\label{prop:rank_condition}
If the data $\{\bm{\rho}_t,\bm{x}_t\}_{t=1}^T$ are consistent with the one-lag habits-over-characteristics model for technology $\bm{A}$, then for every $t\le T-1$,
\[
\tag{NC}
\label{NC}
\text{\rm rank}(\Btil_t \mid \bm{\rho}_t^+) = \text{\rm rank}(\Btil_t) \le \min\{K_t^+,\, J\},
\]
where $\mid$ denotes horizontal concatenation.
\end{proposition}
\end{mdframed}

\medskip
Violation of \eqref{NC} at any single date is sufficient to reject the model, providing a sharp, low-dimensional falsification criterion that can be evaluated period by period. The restriction coincides with the necessary spanning condition under intertemporal separability in \citet{blow_revealed_2008}: habit formation changes the \emph{level} of shadow prices but does not expand the price manifold implied by the mapping from goods to characteristics. As such, \eqref{NC} is a structural constraint driven by the geometry of $\bm{A}$ rather than by the curvature or stability of preferences. Because \eqref{NC} is only necessary, however, it does not guarantee consistency: even when \eqref{NC} holds, the inequalities for goods consumed at zero quantities may still be infeasible. Nevertheless, \eqref{NC} substantially reduces the feasible set and provides a fast diagnostic for empirical implementation.

Beyond its falsification role, the rank restriction also has limited identification content. For any $t$, \eqref{rank_equation} requires the observed price vector $\bm{\rho}_t^+$ to lie in the column space of $\Btil_t$. Because $\bm{B}_t^a$ is formed by rows of $\bm{B}_t$, the augmented matrix $\Btil_t$ has the same column space as $\bm{B}_t'$, so the feasible set of price vectors is at most $J$-dimensional. When habit formation is present ($J_2>0$), the mapping from $(\bm{\pi}_t^0,\bm{\pi}_{t+1}^1)$ to $\bm{\rho}_t^+$ is therefore not one-to-one: reallocating shadow value between contemporaneous and lag components along the habit-forming directions can leave $\bm{\rho}_t^+$ unchanged. Viewed through \eqref{rank_equation} alone, the decomposition into contemporaneous and habit components is thus generically set-identified. In the running cereal example, prices may identify the overall shadow value of ``nuttiness'' and ``sensory intensity,'' but not separately how much of the value of sensory intensity reflects current taste versus its continuation value through habits.

The full characterisation in Theorem~\ref{Afriat} can nevertheless sharpen this conclusion. Condition \eqref{(B1)} links the stacked discounted shadow prices $\tilde{\bm{\pi}}_t$ across dates through cyclical monotonicity evaluated at the observed augmented bundles $\tilde{\bm{z}}_t$, and therefore uses quantity variation as well as prices. These additional restrictions can shrink the admissible set of decompositions relative to the price-span argument embodied in \eqref{NC}. For example, the data may satisfy Theorem~\ref{Afriat} under the dynamic model while violating the static restriction $\bm{\pi}_t^1\equiv \bm{0}$, in which case the zero-habit specification is excluded from the identified set. Absent further structure, however, Theorem~\ref{Afriat} still need not point-identify, or sign, the habit component separately from the contemporaneous one: \eqref{(B1)} disciplines the joint evolution of the stacked vectors $\tilde{\bm{\pi}}_t$, but does not in general undo the non-uniqueness in the decomposition induced by \eqref{rank_equation}.


\vspace{-1em}
\subsection{Missing prices}
\label{sec:missing_prices}

I now relax the assumption that prices are observed for all market goods. In many applications, prices are recorded only for goods that are actually purchased, giving rise to a missing price problem.\footnote{Throughout, I assume that the technology matrix $\bm{A}$ is known and time-invariant. This reflects settings where characteristics can be directly observed or constructed (e.g., nutritional content, design features, emissions ratings), even when market prices are only recorded for purchased items; in practice, missing prices are far more common than missing characteristics.} Missing prices complicate RP analysis because imputing unobserved prices requires auxiliary assumptions. An alternative is to treat missing prices as unknowns and ask whether there exist values that render the data rationalisable. This existence approach should be understood as a partial-identification device: without further restrictions, one can always rationalise non-purchases by assigning prohibitively high unobserved prices, so the goal is to characterise when the observed prices alone already force a violation (or allow rationalisation) under the maintained technology. I now formalise this approach.

Let $\bm{\rho}^+_t$ denote the $K^+_t \leq K$ sub-vector of period-$t$ discounted prices for goods with strictly positive demand, and let $\bm{B}_t$ and $\bm{B}^a_t$ denote the corresponding $J \times K^+_t$ and $J_2 \times K^+_t$ sub-matrices of $\bm{A}$ and $\bm{A}^a$. Let $\bm{\rho}^0_t$, $\bm{B}^0_t$, and $\bm{B}^{a,0}_t$ denote the complementary sub-vectors and sub-matrices associated with zero demand. The full discounted price vector is $\bm{\rho}_t=(\bm{\rho}^+_t,\bm{\rho}^0_t)$. I can then state an Afriat-type characterisation for the habits-over-characteristics model with missing prices.

\vspace{1em}
\begin{mdframed}
\begin{theorem}
\label{Afriat_Missing_Prices}
The following statements are equivalent:
\begin{itemize}
\item[(A$^+$)] There exist missing prices $\{\bm{\rho}^0_t\}_{t=1}^T$ such that the completed price-quantity data are consistent with the model for technology $\bm{A}$.
\item[(B$^+$)] There exist shadow discounted prices $\{\bm{\pi}_t^{0}\}_{t=1}^T$, $\{\bm{\pi}_t^{1}\}_{t=1}^T$ and a discount factor $\beta\in(0,1]$ such that,
\begin{align}
& \label{(B1_missing)}
0 \leq \sum_{m=1}^{M} \pitil_{t_m}^{\prime}\left(\ztil_{t_{m+1}}-\ztil_{t_m}\right)
&&
 \forall \, M \geq 2,
 \forall \, (t_1,\ldots,t_M)\in\{1,\ldots,T\}^{M},\ t_{M+1}=t_1
 \tag{B1$^+$} \\
& \label{(B2_missing)}
\bm{\rho}^+_t
= 
\bm{B}_t^{\prime} \bm{\pi}_t^0 +
(\bm{B}^{a}_t)^{\prime}  \bm{\pi}_{t+1}^{1}
&&
 \forall \, t \in\{1, \ldots, T\} \, 
 \tag{B2$^+$}
\end{align}
where 
$\pitil_t := \frac{1}{\beta^{t-1}}\left[\bm{\pi}_t^{0 \prime}, \bm{\pi}_t^{1 \prime}\right]^{\prime}$, with $\bm{\pi}_{T+1}^1 \equiv \bm{0}$.
\end{itemize}
\end{theorem}
\end{mdframed}

\textit{Proof:} See Appendix~\ref{sec:proofs}.  $\qed$

\medskip
Relative to Theorem~\ref{Afriat}, the missing-price test conditions only on purchased-good prices and is therefore weaker. In empirical settings where the initial stock is unobserved and the observed sample window is not treated as the consumer's terminal horizon, the implementable counterpart is the observable interior-date analogue of Theorem~\ref{Afriat_Missing_Prices}, which evaluates the restrictions on dates for which both the lagged bundle and the one-step-ahead continuation term are observed.

\vspace{-1em}
\section{Corollaries and nesting results}

Before turning to the empirical application, I record two immediate corollaries that situate the habits-over-characteristics model within the RP literature. First, when characteristics coincide with market goods, the framework collapses to the habits-over-goods model of \citet{crawford_habits_2010}. Second, under intertemporal separability and exponential discounting, it reduces to the characteristics-based model of \citet{blow_revealed_2008}. Both results follow directly from Theorem~\ref{Afriat} once the model is specialised to the relevant limiting cases. Formal definitions and derivations for these special cases are provided in Appendix~\ref{sec:proofs}.

\vspace{-1em}
\subsection{Habits-over-goods as a special case}
\vspace{-0.5em}

When characteristics coincide with market goods---that is, when $J = K$ and the technology matrix satisfies $\bm{A} = \bm{I}_J$---the habits-over-characteristics framework reduces to a standard habits-over-goods model. In this case, the distinction between goods and characteristics disappears, and the intertemporal first-order conditions involve current and lagged consumption of habit-forming goods directly. To match \citet{crawford_habits_2010}, I maintain the additional assumption that all goods are consumed in strictly positive quantities.
Throughout this subsection, I adopt the terminal convention $\bm{\rho}_{T+1}^{a,1}\equiv \bm{0}$, so the final-period habit-good restriction has no continuation term. I also normalise the marginal utility of lifetime wealth to one, $\lambda=1$, without loss of generality.

\vspace{0.15em}
\begin{mdframed}
\begin{definition}[]
\label{consistency_goods}
The data $\left\{\bm{\rho}_t^c, \bm{\rho}_t^a; \bm{x}_t^c, \bm{x}_t^a\right\}_{t=1}^T$ are \textit{consistent} with the one-lag habits-over-goods model if there exists a locally non-satiated, superdifferentiable, and concave utility function $u(\cdot)$ and a positive constant $\beta$ such that for all $t \in \{1, \ldots, T\}$:
\vspace{-0.5em}
\begin{align}
\label{foc_market_prices_goods}
\bm{\rho}_t^c 
& =
\beta^{t-1}
\partial_{\bm{x}_t^c} u(\bar{\bm{x}}_t), \\
\bm{\rho}_t^a 
& =
\beta^{t-1}
\partial_{\bm{x}_t^a} u(\bar{\bm{x}}_t) 
+
\beta^{t}
\partial_{\bm{x}_t^a} u(\bar{\bm{x}}_{t+1}),
\vspace{-0.5em}
\end{align}
\vspace{-0.25em}
where $\bm{\rho}_t^{c}$ and $\bm{\rho}_t^{a}$ denote present-value prices of non-habit-forming and habit-forming goods, respectively.
\end{definition}
\end{mdframed}
\vspace{0.15em}

Given this definition, I obtain the following result, equivalent to that in \cite{crawford_habits_2010}.\footnote{Equivalence uses the normalisation of the marginal utility of lifetime wealth $\lambda = 1$, which can be imposed without loss of generality.}

\vspace{0.15em}
\begin{mdframed}
\begin{corollary}[]
\label{AfriatGoods}
The following statements are equivalent:
\vspace{-0.5em}
\begin{itemize}
\item[(A')] 
The data $\left\{\bm{\rho}_t^c, \bm{\rho}_t^a; \bm{x}_t^c, \bm{x}_t^a\right\}_{t=1}^T$ are consistent with the one-lag habits-over-goods model.

\item[(B')]
There exist $T$ shadow price vectors $\left\{\bm{\rho}_t^{a,0}\right\}_{t=1}^{T}$, $T$ shadow price vectors $\left\{\bm{\rho}_t^{a,1}\right\}_{t=1}^{T}$, and a positive constant $\beta$ such that:
\vspace{-0.5em}
\begin{align}
\label{(B1_goods)}
0 & \leq \sum_{m=1}^{M}\rhotil_{t_m}^{\prime}(\bar{\bm{x}}_{t_{m+1}}-\bar{\bm{x}}_{t_m})
&& \forall \, M \geq 2,\ \forall \, (t_1,\ldots,t_M)\in\{1,\ldots,T\}^{M},\ t_{M+1}=t_1, \tag{B1'} \\
\label{(B3_goods)}
\rho_t^{a, k} & = \bm{e}_k \bm{\rho}_t^{a, 0} + \bm{e}_k \bm{\rho}_{t+1}^{a, 1} && \text{if } x_t^k > 0, \quad \forall k, t \in \{1, \ldots, T\}, \tag{B3'}
\end{align}
\vspace{-0.5em}
where $\bar{\bm{x}}_t := (\bm{x}_t^{c \prime}, \bm{x}_t^{a \prime}, \bm{x}_{t-1}^{a \prime})^{\prime}$, $\bm{e}_k$ is the $k$-th standard basis vector, and
$
\rhotil_t := 
\frac{1}{\beta^{t-1}}
\begin{bmatrix}
\bm{\rho}_t^{c \prime}, \bm{\rho}_t^{a, 0 \prime}, \bm{\rho}_t^{a, 1 \prime}
\end{bmatrix}^{\prime}.
$
\end{itemize}
\end{corollary}
\end{mdframed}

\noindent
\textit{Proof:} See Appendix~\ref{sec:proofs}.  $\qed$
\vspace{0.75em}

\vspace{-1em}
\subsection{Intertemporally separable preferences over characteristics}
\vspace{-0.25em}

If habit formation is absent so that $\bm{z}_t = \bm{z}_t^c$ for all $t$, the model reduces to a characteristics-based framework with intertemporally separable preferences. Unlike \citet{blow_revealed_2008}, who remain agnostic about intertemporal allocation, my formulation embeds this static characteristics model within a lifecycle problem with exponential discounting.\footnote{By a lifecycle problem I mean that the consumer chooses the entire consumption path to maximise lifetime utility subject to a single present-value budget constraint under exponential discounting.} The presence of a single intertemporal budget constraint implies a single shadow value of lifetime wealth, so observed choices must satisfy both intratemporal utility maximisation over characteristics and intertemporal optimality. Consistency with this lifecycle formulation is defined as follows.
As in the rest of the paper, I normalise the marginal utility of lifetime wealth to one.

\vspace{0.25em}
\begin{mdframed}
\begin{definition}[]
\label{consistency_sep}
The data $\left\{\bm{\rho}_t; \bm{x}_t\right\}_{t=1}^T$ are \textit{consistent} with intertemporally separable preferences over characteristics and a life-cycle model for given technology $\bm{A}$ if there exists a locally non-satiated, superdifferentiable, and concave utility function $u(\cdot)$ such that, for all $t \in \{1, \ldots, T\}$,
\vspace{-0.5em}
\begin{align}
\label{foc_market_prices_sep}
\bm{\rho}_t 
\geq    
\bm{A}^{\prime} \bm{\pi}_t,
\end{align}
\vspace{-0.5em}
with equality for all $k$ such that $x_t^k > 0$, where $\bm{z}_t = \bm{A} \bm{x}_t$, $\bm{\rho}_t$ denotes the vector of present-value prices, and
$
\bm{\pi}_t := \partial_{\bm{z}_t} u(\bm{z}_t).
$
\end{definition}
\end{mdframed}

This recovers the hedonic pricing equation of \citet{gorman_possible_1956} and the first-order condition of \citet{blow_revealed_2008} when the static characteristics model is embedded in a lifecycle framework. From this, I obtain the following characterisation.

\vspace{-0.25em}
\begin{mdframed}[innertopmargin=3pt,innerbottommargin=7pt]
\small
\begin{corollary}[]
\label{AfriatSep}
The following are equivalent:
\vspace{-0.5em}
\begin{itemize}
\item[(A'')] 
The data $\left\{\bm{\rho}_t; \bm{x}_t\right\}_{t=1}^T$ are consistent with intertemporally separable preferences over characteristics and a life-cycle model.
\vspace{-0.5em}
\item[(B'')]
There exist $T$ shadow price vectors $\left\{\bm{\pi}_t\right\}_{t=1}^T$ such that:
\vspace{-0.5em}
\begin{align}
\label{(B1sep)}
0 & \leq \sum_{m=1}^{M}\bm{\pi}_{t_m}^{\prime}(\bm{z}_{t_{m+1}}-\bm{z}_{t_m})
&& \forall \, M \geq 2,\ \forall \, (t_1,\ldots,t_M)\in\{1,\ldots,T\}^{M},\ t_{M+1}=t_1, \tag{B1''} \\
\label{(B2sep)}
\rho^k_t & \geq \bm{a}_k^{\prime} \bm{\pi}_t && \forall \, k, t \in \{1, \ldots, T\}, \tag{B2''} \\
\label{(B3sep)}
\rho^k_t & = \bm{a}_k^{\prime} \bm{\pi}_t && \text{if } x_t^k > 0, \quad \forall \, k, t \in \{1, \ldots, T\}, \tag{B3''}
\end{align}
where $\bm{a}_k$ is the $k$th column of $\bm{A}$.
\end{itemize}
\end{corollary}
\end{mdframed}

\noindent
\textit{Proof:} See Appendix~\ref{sec:proofs}.  $\qed$

Together, these corollaries clarify the paper's value added. The framework unifies existing RP results in goods space and in characteristics space within a single dynamic hedonic model, and it shows how the interpretation of prices changes once habit formation over attributes is allowed. Relative to intertemporal separability, the dynamic model introduces continuation values into the shadow-price system and delivers diagnostics---via feasibility of \eqref{(B1)}--\eqref{(B3)} and the rank condition \eqref{NC}---that separate failures of the hedonic price representation from failures of dynamically coherent preferences.

\vspace{-1em}
\section{Testing dynamic valuation in cereal purchases}

This section applies the framework to household-level cereal purchases as a diagnostic of when characteristic-based valuation is admissible in scanner data, and when it requires a dynamic rather than static interpretation. Using these data, I show that most differences in raw rationalisability across models are driven by the \emph{structural} restrictions imposed by the hedonic representation of observed prices. Conditional on those restrictions, allowing for habits improves behavioural coherence, though the gains are more visible in the severity of violations than in binary pass rates. The goal of the application is not to estimate a cereal-demand model, but to illustrate the framework in a realistic scanner setting where only purchases are observed and the mapping from purchases to underlying consumption states must be treated as a maintained approximation rather than as a directly observed object.

The empirical exercise proceeds sequentially. First, I ask a structural question: do observed within-period price vectors admit a hedonic representation? Equivalently, do they satisfy the equalities implied by the goods-to-characteristics technology (i.e., (B2$^+$) in Theorem \ref{Afriat_Missing_Prices}), so that characteristic shadow prices are well defined? These equalities act as a gatekeeper: if they fail, the behavioural test is not economically meaningful because the relevant shadow prices do not exist. Second, conditional on structural feasibility, I ask a behavioural question: do observed purchase sequences satisfy the dynamic RP inequalities (i.e., (B1$^+$)) when evaluated at the implied shadow prices? Put differently, once prices can be written as characteristic shadow values, I ask whether the household's sequence of purchases is consistent with stable dynamic preferences over those characteristics. In the data, these inequalities are applied to purchase-period bundles, so the resulting exercise should be read as a diagnostic of dynamic coherence under the maintained approximation that purchases track the relevant underlying consumption state at the chosen aggregation.

The sequential decomposition yields two findings. First, most variation in raw pass rates when moving from goods space to characteristics space reflects the hedonic price-system restrictions rather than differences in behavioural fit; empirically, however, restoring hedonic consistency typically requires only modest price adjustments. Second, conditional on satisfying those structural restrictions, allowing for habits improves behavioural coherence. The binary gains are concentrated, but the continuous violation measures show that dynamic valuation generally moves behaviour closer to rationalisability. These results should therefore be read as evidence about which restrictions bind in scanner data, and when static valuation remains a useful approximation.

Market goods are UPC-level products and characteristics are their nutritional content and a small set of descriptive indicators. The key empirical challenge is that prices are observed only for purchased items. Accordingly, I implement the missing-price test of Section~\ref{sec:missing_prices}.\footnote{\label{footnote:impute}Imputing prices for unpurchased goods is feasible (e.g., using regional price indices), but it introduces auxiliary assumptions that can blur whether failures reflect preferences or the imputation procedure.} In this setting, behavioural discipline is inherently limited by sparse price support and modest intertemporal budget variation, so the diagnostic provides a lower bound on the model's behavioural content. In environments with richer price and quantity support, the same framework may generate sharper behavioural restrictions.

These scanner-data features also affect how failures should be interpreted. Unmeasured product attributes, temporary promotions, inventory behaviour, and purchase--consumption mismatch can all generate violations of an exact RP model. I therefore use pass rates together with distance measures to evaluate whether a maintained valuation exercise is a useful approximation, rather than treating failure as literal evidence against rational choice.

\vspace{-1em}
\subsection{Data}
\label{sec:the_data}

I use household-level scanner data on cold cereal purchases from the IRI Academic Datasets' \textit{BehaviorScan} panel. The panel covers two U.S. markets (Pittsfield, MA, and Eau Claire, WI) over 2010--2011, with purchases recorded at checkout via household ID cards and Universal Product Codes (UPCs). I restrict attention to static households that participate in all 12 months of a calendar year, so recruitment and attrition occur only at year-end.

I compute household-specific time periods to accommodate heterogeneity in purchase frequencies. For household $i$, let $S_i$ denote the span from first to last purchase and $G_i$ the longest interpurchase gap (including endpoints). I set $T_i=\lfloor S_i/G_i\rfloor$ and partition $S_i$ into $T_i$ equal-length bins, which guarantees at least one purchase in each period by construction. Households with $T_i<3$ are excluded to ensure at least two observed transitions in the one-lag model. This aggregation is conservative for a one-lag specification: it ensures that the lagged bundle is observed rather than imputed, hence dynamic restrictions are evaluated on realised purchase transitions. After additionally dropping households with purchases lacking characteristics information (described below), the analysis sample contains $N=2{,}282$ households.\footnote{Appendix~\ref{app:data_appendix} documents the full sequence of sample construction and reports balance tests comparing households excluded due to missing characteristics data with the final analysis sample. Differences are modest in magnitude, suggesting limited scope for selection on observables.} Table~\ref{tab:data_summary} summarises the resulting panel structure and the scale and variety of the implied choice problems.

\begin{table}[!h]
\centering
\caption{\label{tab:data_summary}
Sample and household-level panel structure}
\fontsize{10}{12}\selectfont
\begin{tabular}[t]{>{\raggedright\arraybackslash}p{6.4cm}>{\raggedright\arraybackslash}p{3cm}>{\raggedright\arraybackslash}p{3cm}}
\toprule
Statistic & Mean & Median \\
\midrule 
\rowcolor{gray!10} Households ($N$) & \multicolumn{2}{l}{2{,}282}  \\
Months covered & \multicolumn{2}{l}{2010--2011 (24 months)} \\
\rowcolor{gray!10} Distinct products ($K$) & \multicolumn{2}{l}{801} \\
Characteristics ($J$) & \multicolumn{2}{l}{23} \\
\rowcolor{gray!10} Time periods per household ($T_i$) & 7.0 & 6.0 \\
Period length (days) & 106 & 98 \\
\rowcolor{gray!10} Units purchased per household-period & 7.9 & 6.0 \\
Total units purchased (24 months) & 55.1 & 42.0 \\
\rowcolor{gray!10} Expenditure per household-period (\$) & 22.1 & 19.5 \\
Total expenditure (24 months, \$) & 154.2 & 116.5 \\
\rowcolor{gray!10} Distinct products purchased (24 months) & 32.3 & 25.0 \\
\bottomrule
\end{tabular}
\footnotesize
\begin{flushleft}
\textit{Notes:} Statistics are computed across $N=2{,}282$ households. $T_i$ denotes the number of constructed time periods for household $i$, defined as $T_i=\lfloor S_i/G_i \rfloor$, where $S_i$ is the span between first and last purchase and $G_i$ the longest interpurchase gap.
Units are package counts; expenditures and package counts are aggregated within household-periods and totals are over 24 months.  
\end{flushleft}
\end{table}
\vspace{-2em}

Households purchase on average 7.9 units per period (median 6.0). Combined with a median period length of 98 days, this pattern is consistent with ongoing consumption rather than extreme purchase spikes. While promotions may induce some stockpiling, explaining the observed interpurchase gaps purely through inventory accumulation would require implausibly large stock build-ups relative to total quantities purchased. At the same time, within-period baskets remain narrow in UPC variety: the median household-period contains only two distinct purchased products, even though households typically buy several units. I therefore interpret each household-period as a meaningful dynamic choice observation, but not as a literal measure of contemporaneous nutrient intake: it accommodates idiosyncratic shopping frequencies while preserving the temporal structure required for testing whether aggregated purchase bundles behave in a manner consistent with the dynamic RP restrictions.\footnote{Similar coarse aggregation is common in empirical work on dynamic demand when purchase occasions are intermittent (e.g., \citet{crawford_habits_2010} uses quarterly panels in an application to tobacco). I obtain qualitatively similar cross-model comparisons under a common monthly aggregation (i.e., $T_i=24$ for all $i$), albeit in a much smaller effective sample because zero-purchase months become pervasive.}

Figure~\ref{fig:data_characteristics} shows substantial heterogeneity in shopping frequency (median $T_i=6$; median period length 98 days). Within each period I aggregate UPC-level quantities and nominal expenditures and compute unit values as expenditure-to-quantity ratios.\footnote{Expenditures are recorded in nominal dollars and converted to present-value terms for the lifecycle formulation using a monthly interest-rate series (30-Year Fixed Rate Mortgage Average, FRED \citep{fred_mortgage30us}); results are unchanged under nominal values given the short sample window.} Because unpurchased goods have no observed unit values, the test treats missing prices as unknowns and searches for completions consistent with rationalisability. Reported pass rates are therefore upper bounds: a household that fails cannot be rescued by any price imputation, while a household that passes does so under at least one completion.

\FloatBarrier
\begin{figure}[!h]
\caption{Distribution of constructed time periods and period lengths}
\label{fig:data_characteristics}
\centering
\begin{subfigure}[b]{0.49\textwidth}
\centering
\includegraphics[width=\textwidth]{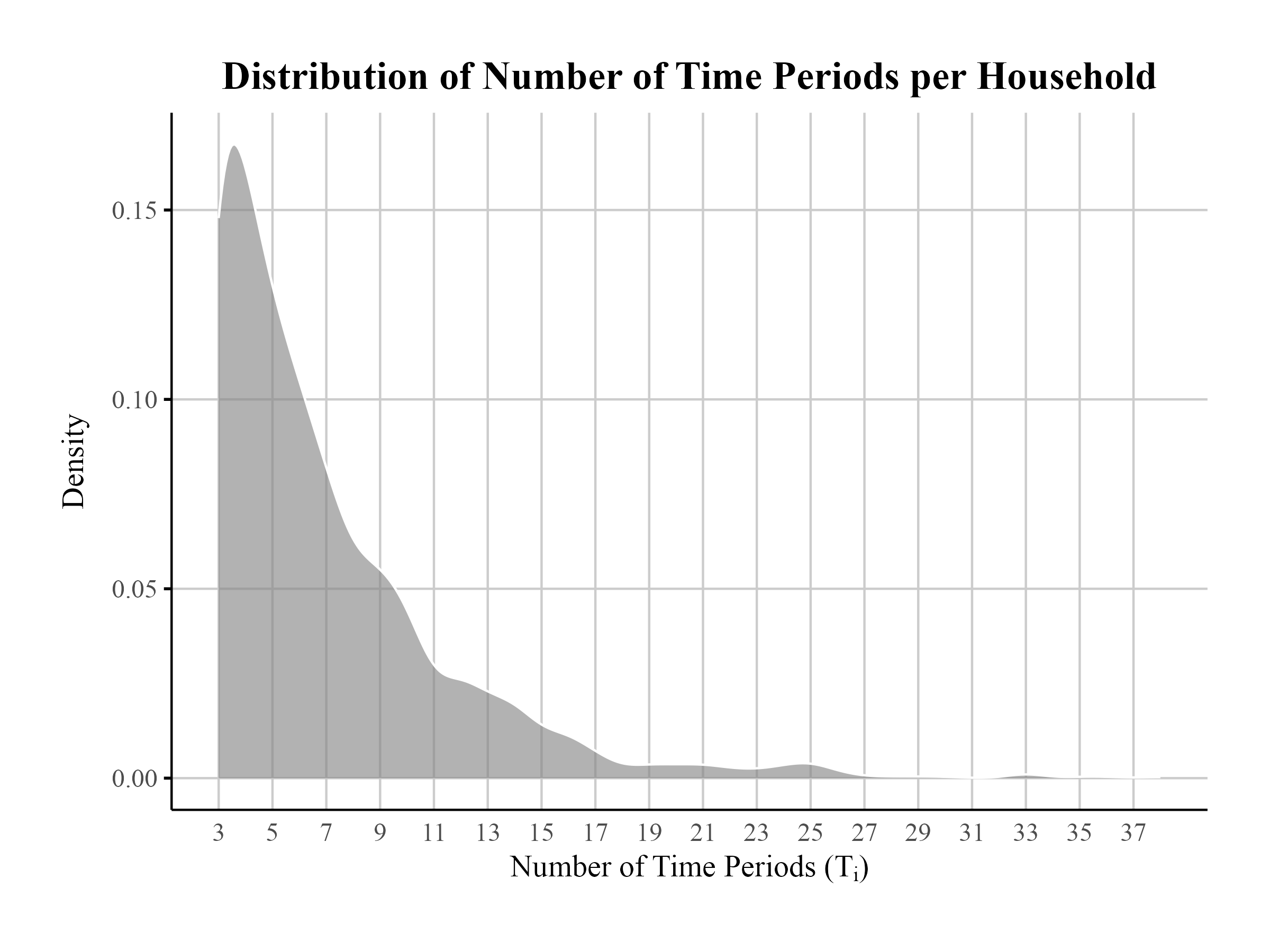}
\caption{Distribution of $T_i$ across households}
\label{fig:num_time_periods_density}
\end{subfigure}
\hfill
\begin{subfigure}[b]{0.49\textwidth}
\centering
\includegraphics[width=\textwidth]{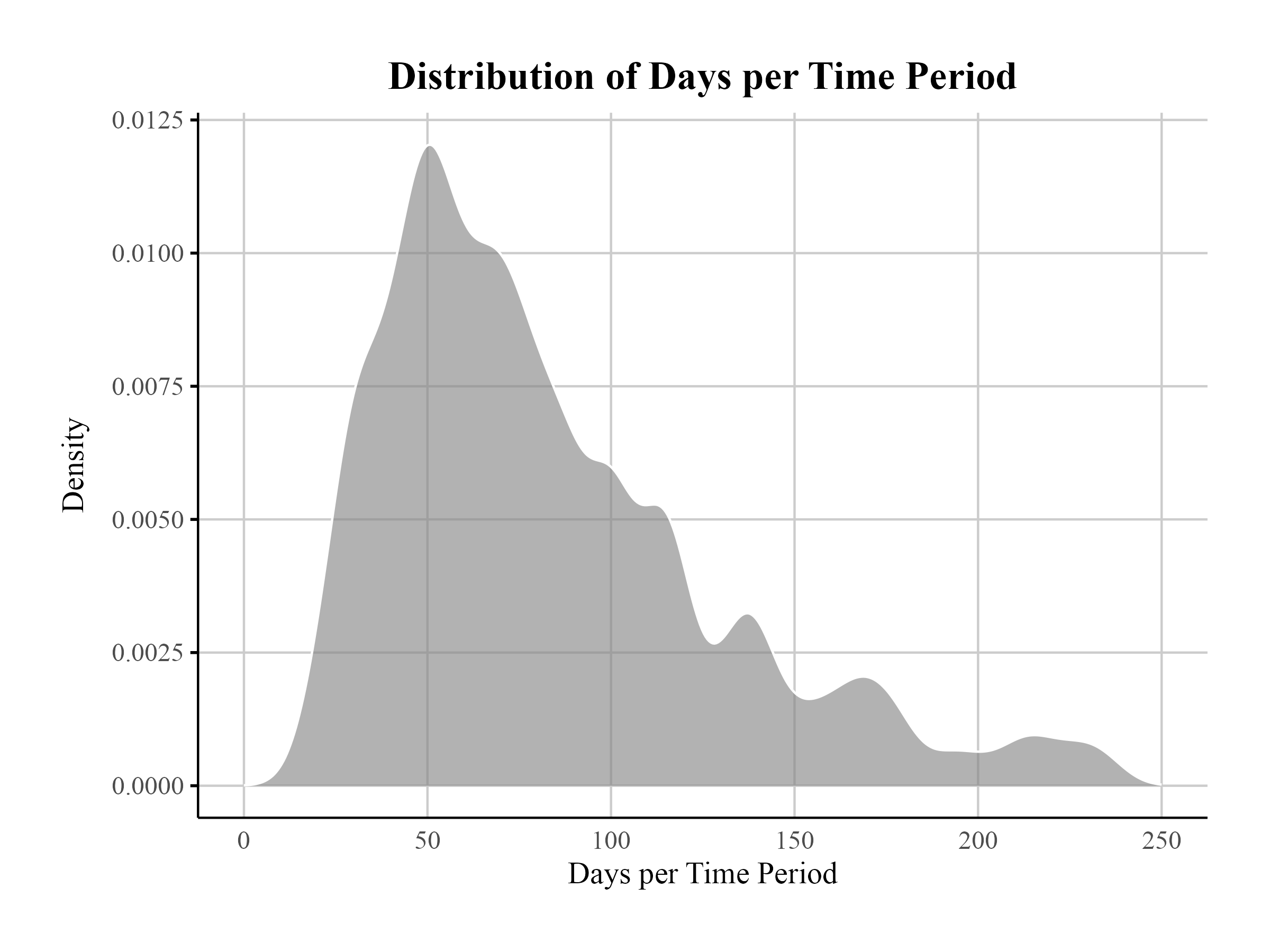}
\caption{Distribution of period lengths (days)}
\label{fig:days_per_time_period_density}
\end{subfigure}
\footnotesize
\begin{flushleft}
\textit{Notes:} Left panel shows the distribution of $T_i$ across households ($N=2{,}282$). 
Right panel shows the distribution of household-level average period length in days. 
$T_i=\lfloor S_i/G_i \rfloor$, where $S_i$ is the span between first and last purchase and $G_i$ the longest interpurchase gap.
\end{flushleft}
\end{figure}
\FloatBarrier
\vspace{-2em}

The model is defined over characteristics rather than market goods. For each UPC I construct a vector of nutritional and descriptive attributes by merging the IRI product file to the (now-defunct) NuVal shelf-labelling database and, where NuVal is missing, supplementing with data from the FatSecret Platform API, following the approach in \citet{barahona_equilibrium_2023}. I standardise all nutrients to a 100g basis so that continuous characteristics are comparable across products. Since scanner quantities are recorded in package counts, the resulting characteristics matrix should be interpreted as a linear summary of purchased product attributes. To handle IRI's placeholder codes for some private-label items, I first map System-88 pseudo-UPCs to their corresponding real UPCs prior to merging. The merged characteristics cover over 97\% of purchase-weighted observations; I drop unmatched purchase records and then drop households with any remaining unmatched purchases, so every household in the analysis sample has a complete characteristics mapping for all of its observed purchases.

I construct $J=23$ characteristics. Eight are continuous Nutrition Facts Panel measures (calories, carbohydrates, total fat, saturated fat, fibre, protein, sodium, and sugar). The remaining 15 are binary indicators capturing salient ingredients and descriptors (10 indicators) and the five most prevalent brands (Kellogg's, General Mills, Post, Quaker, and Kashi). Table~\ref{tab:binary_char} reports the binary definitions. In the baseline specification I treat sugar and sodium as habit-forming characteristics ($J_2=2$) and take the remaining characteristics as non-habit-forming, motivated by evidence that sugar and salt may activate reward pathways in ways analogous to addictive substances \citep{avena_evidence_2008, cocores_salted_2009}. Section~\ref{sec:rat_scores} reports robustness to alternative partitions. Additional descriptive evidence on purchase intensity, brand concentration, and the distributions of prices and characteristics is reported in Appendix~\ref{app:data_appendix}.

\begin{table}[!h]
\centering
\caption{Binary product characteristics used in the hedonic representation}
\label{tab:binary_char}
\fontsize{10}{12}\selectfont
\begin{tabular}[t]{>{\raggedright\arraybackslash}p{3.5cm}>{\raggedright\arraybackslash}p{12cm}}
\toprule
Binary characteristic & Construction \\
\midrule
\cellcolor{gray!10}{whole grain} & \cellcolor{gray!10}{contains `whole grain' or `wholegrain' in ingredient list} \\
organic & contains the term `organic' at least once in ingredient list \\
\cellcolor{gray!10}{oat-based} & \cellcolor{gray!10}{contains `oat' as part of the first listed ingredient} \\
granola & contains `granola' in product description \\
\cellcolor{gray!10}{health halos} & \cellcolor{gray!10}{contains `source of', `natural', `low calorie', or `low fat' in product description} \\
fruity & contains `fruit' in the product description \\
\cellcolor{gray!10}{nutty} & \cellcolor{gray!10}{contains `nut' in product description} \\
chocolatey & contains `chocolate' in product description \\
\cellcolor{gray!10}{honey flavour} & \cellcolor{gray!10}{contains `honey' in product description} \\
gluten-free & contains `gluten free' in product description \\
\cellcolor{gray!10}{Kellogg's} & \cellcolor{gray!10}{cereal brand is Kellogg's} \\
General Mills & cereal brand is General Mills \\
\cellcolor{gray!10}{Post} & \cellcolor{gray!10}{cereal brand is Post} \\
Quaker & cereal brand is Quaker \\
\cellcolor{gray!10}{Kashi} & \cellcolor{gray!10}{cereal brand is Kashi} \\
\bottomrule
\end{tabular}
\footnotesize
\begin{flushleft}
\textit{Notes:} Characteristics are defined at the UPC level using ingredient lists and product descriptions from the merged IRI--NuVal--FatSecret dataset. 
Indicators equal one if the stated textual condition is satisfied. 
Brand indicators correspond to the five most prevalent national brands in the sample. 
Continuous nutritional characteristics (calories, carbohydrates, total fat, saturated fat, fibre, protein, sodium, and sugar, standardised per 100g) are defined separately in Appendix~\ref{app:data_appendix}.
\end{flushleft}
\end{table}

 
\vspace{-3em}
\subsection{Results}
\label{sec:results}

\vspace{-0.5em}
\subsubsection{Rationalisability scores}
\label{sec:rat_scores}

Two patterns emerge in the raw rationalisability outcomes: allowing for habits increases pass rates within characteristics space, and goods-based representations exhibit substantially higher pass rates. As the sequential logic of the paper makes clear, however, these raw differences conflate structural and behavioural components. Section~\ref{sec:restrictiveness} decomposes these margins and quantifies the severity of violations.

The test is implemented at the household level, allowing full heterogeneity in pass/fail outcomes and in the shadow-price structure of rationalisable households. Because habits enter with a one-period lag, the empirical exercise uses the observable interior-date analogue of Theorem~\ref{Afriat_Missing_Prices}. Structural equalities are imposed on $t=2,\ldots,T-1$, since period $t=1$ depends on the unobserved initial stock and I do not treat the last observed purchase period as the consumer's terminal horizon. The behavioural Afriat inequalities are then evaluated on the retained dates $t=2,\ldots,T$, so the final observed period continues to discipline the shadow-price sequence. Unless otherwise specified, I impose a lifecycle model with a single intertemporal budget constraint, so the marginal utility of income is constant across periods.

Under the baseline habits-over-characteristics specification, 1{,}248 of 2{,}282 households (54.69\%) satisfy the test. Varying the set of habit-forming characteristics has essentially no effect on classification: restricting habits to sugar alone (54.60\%), sodium alone (54.65\%), or allowing all 23 characteristics to be habit-forming (54.69\%) changes the pass rate by at most 0.09 percentage points and reclassifies no more than two households. In this application, the data therefore do not isolate sugar or sodium as uniquely responsible for the dynamic improvement; rather, the main empirical action comes from the maintained hedonic representation, with habits mattering at the behavioural margin once that representation is imposed. Table~\ref{tab:main_results} reports the full set of specifications.

\begin{table}[!h]
\centering
\caption{\label{tab:main_results}
Household-level pass rates: habits-over-characteristics specifications}
\fontsize{10}{12}\selectfont
\begin{tabular}[t]{>{\raggedright\arraybackslash}p{5.5cm}>{\raggedright\arraybackslash}p{5cm}>{\raggedright\arraybackslash}p{2.5cm}>{\raggedleft\arraybackslash}p{2cm}}
\toprule
Model & Description & Lifecycle model & Pass rate \\
\midrule
\cellcolor{gray!10}{Habits-over-characteristics} & \cellcolor{gray!10}{$J_2 = 2$ (sugar and sodium)} & \cellcolor{gray!10}{Yes} & \cellcolor{gray!10}{54.69\%} \\
Habits-over-sugar & $J_2 = 1$ (sugar) & Yes & 54.60\% \\
\cellcolor{gray!10}{Habits-over-sodium} & \cellcolor{gray!10}{$J_2 = 1$ (sodium)} & \cellcolor{gray!10}{Yes} & \cellcolor{gray!10}{54.65\%} \\
Habits-over-all-characteristics & $J_2 = J = 23$ & Yes & 54.69\% \\
\bottomrule
\end{tabular}
\footnotesize
\begin{flushleft}
\textit{Notes:} A household passes if there exist characteristic shadow prices satisfying the interior-date structural equalities (B2$^{+}$) in Theorem~\ref{Afriat_Missing_Prices} on $t=2,\ldots,T-1$ and the corresponding dynamic RP conditions (B1$^{+}$) on $t=2,\ldots,T$; by Theorem~\ref{Afriat_Missing_Prices}, this is equivalent to the existence of a valid completion of missing prices in the background. 
Pass rates are computed over $N=2{,}282$ households. 
$J=23$ denotes the total number of characteristics and $J_2$ the number assumed to be habit-forming. 
All specifications impose a lifecycle model with a single intertemporal budget constraint.
\end{flushleft}
\end{table}
\vspace{-2em}

Raw pass rates differ sharply across representations (Table~\ref{tab:comparison_model_results}). Moving from characteristics to goods sharply increases rationalisability from 54.7\% to 99.6\% under dynamic preferences. This large difference is driven by the much greater dimensional flexibility of goods-based representations, which impose no cross-good price restrictions. One could mechanically raise characteristics-space pass rates by adding latent product attributes to the technology matrix. This plays a role analogous to absorbing residual variation into unobserved product characteristics, but it would change the object being tested: whether an enlarged characteristic space rationalises choices, rather than whether observed policy-relevant characteristics alone support the valuation interpretation. Holding fixed the observed characteristics, the structural restrictions bite through the geometry of the active purchased bundle rather than through broad within-period variety. Many of the apparent ``failures'' in characteristics space nevertheless correspond to economically small deviations from those hedonic restrictions. Section~\ref{sec:restrictiveness} makes this precise by separating structural and behavioural sources of empirical discipline and quantifying the magnitude of violations.

By contrast, removing habits within characteristics space reduces the pass rate from 54.7\% to 52.4\%. While this difference in levels is modest, the paired comparisons below show that the reclassification is entirely directional, with households failing under static preferences but passing once dynamics are introduced. The next subsection shows that this modest binary gain is only part of the evidence: allowing for habits also reduces the severity of behavioural violations on average.

\begin{table}[!h]
\centering
\caption{\label{tab:comparison_model_results}
Household-level pass rates across characteristics and goods representations}
\fontsize{10}{12}\selectfont
\begin{tabular}[t]{>{\raggedright\arraybackslash}p{5.5cm}>{\raggedright\arraybackslash}p{5cm}>{\raggedright\arraybackslash}p{2.5cm}>{\raggedleft\arraybackslash}p{2cm}}
\toprule
Model & Description & Lifecycle model & Pass rate \\
\midrule
\cellcolor{gray!10}{Habits-over-characteristics} & \cellcolor{gray!10}{$J_2 = 2$ (sugar and sodium)} & \cellcolor{gray!10}{Yes} & \cellcolor{gray!10}{54.69\%} \\
Characteristics (no habits) & $J_2= 0$ & Yes & 52.37\% \\
\cellcolor{gray!10}{Habits-over-all-goods} & \cellcolor{gray!10}{$K = J = J_2$, $\bm{A} =$ identity($K$)} & \cellcolor{gray!10}{Yes} & \cellcolor{gray!10}{99.56\%} \\
Goods (no habits) & $K = J = J_1$, $\bm{A} =$ identity($K$) & Yes & 92.59\% \\
\cellcolor{gray!10}{Goods (no habits) (GARP)} & \cellcolor{gray!10}{$K = J = J_1$, $\bm{A} =$ identity($K$)} & \cellcolor{gray!10}{No} & \cellcolor{gray!10}{99.69\%} \\
\bottomrule
\end{tabular}
\footnotesize
\begin{flushleft}
\textit{Notes:} A household passes if there exist characteristic shadow prices satisfying the interior-date structural equalities (B2$^{+}$) in Theorem~\ref{Afriat_Missing_Prices} on $t=2,\ldots,T-1$ and the corresponding dynamic RP conditions (B1$^{+}$) on $t=2,\ldots,T$; by Theorem~\ref{Afriat_Missing_Prices}, this is equivalent to the existence of a valid completion of missing prices in the background. 
Pass rates are computed over $N=2{,}282$ households. 
In characteristics models, $J=23$ and $J_2$ denotes the number of habit-forming characteristics. 
In goods models, $K$ denotes the number of goods and $\bm{A}=\mathrm{identity}(K)$ implies no cross-good price restrictions. 
``Lifecycle model'' indicates whether a single intertemporal budget constraint is imposed.
\end{flushleft}
\vspace{-2em}
\end{table}

The reclassification pattern is strongly directional. Removing habits from the baseline characteristics model reclassifies 53 households, all of whom fail under static preferences but pass once intertemporal dependence is introduced ($\text{p-value}<10^{-15}$). These households provide the cleanest empirical illustration of the valuation wedge in Section~\ref{sec:valuation_wedge}: the maintained characteristics technology can support a hedonic interpretation, but only if current characteristics are allowed to carry continuation values. By contrast, changing the allocation of habit-forming characteristics reclassifies at most three households and yields no statistically significant differences. Table~\ref{tab:mcnemar_results} reports the full set of paired comparisons.

Comparisons with goods-based models reveal even larger directional differences. In these cases, most switching households fail the characteristics-based test but pass the corresponding goods-based alternative. As shown in the next subsection, this pattern reflects the much greater dimensional flexibility of the goods representation, which imposes no cross-good price restrictions, rather than tighter behavioural alignment.

\begin{table}[!h]
\centering
\caption{\label{tab:mcnemar_results}Paired model comparisons (exact McNemar tests)}
\fontsize{10}{12}\selectfont
\begin{tabular}[t]{>{\raggedright\arraybackslash}p{6.4cm}>{\raggedleft\arraybackslash}p{1.6cm}>{\raggedleft\arraybackslash}p{1.6cm}>{\raggedleft\arraybackslash}p{1.6cm}>{\raggedleft\arraybackslash}p{1.6cm}>{\raggedleft\arraybackslash}p{1.6cm}}
\toprule
Comparison (baseline vs alternative) & Pass$_0$ & Pass$_1$ & $\Delta$ (pp) & Switchers & $p$-value \\
\midrule
\cellcolor{gray!10}{Baseline: Habits-over-characteristics} & \cellcolor{gray!10}{54.69} & \cellcolor{gray!10}{} & \cellcolor{gray!10}{} & \cellcolor{gray!10}{} & \cellcolor{gray!10}{} \\
Habits-over-all-characteristics & 54.69 & 54.69 & 0.00 & 0 & 1.000 \\
\cellcolor{gray!10}{Habits-over-sugar} & \cellcolor{gray!10}{54.69} & \cellcolor{gray!10}{54.60} & \cellcolor{gray!10}{0.09} & \cellcolor{gray!10}{2} & \cellcolor{gray!10}{0.500} \\
Habits-over-sodium & 54.69 & 54.65 & 0.04 & 1 & 1.000 \\
\cellcolor{gray!10}{Characteristics (no habits)} & \cellcolor{gray!10}{54.69} & \cellcolor{gray!10}{52.37} & \cellcolor{gray!10}{2.32} & \cellcolor{gray!10}{53} & \cellcolor{gray!10}{$<10^{-15}$} \\
Habits-over-all-goods & 54.69 & 99.56 & -44.87 & 1024 & $<10^{-15}$ \\
\cellcolor{gray!10}{Goods (no habits)} & \cellcolor{gray!10}{54.69} & \cellcolor{gray!10}{92.59} & \cellcolor{gray!10}{-37.91} & \cellcolor{gray!10}{1039} & \cellcolor{gray!10}{$<10^{-15}$} \\
Goods (no habits) (GARP) & 54.69 & 99.69 & -45.00 & 1031 & $<10^{-15}$ \\
\bottomrule
\end{tabular}
 
\vspace{0.2cm}
\footnotesize
\textit{Notes:} Pass$_0$ is the baseline pass rate (habits-over-characteristics with $J_2=2$), Pass$_1$ is the alternative model pass rate, and $\Delta$ (pp) reports Pass$_0$ minus Pass$_1$ in percentage points. ``Switchers'' is $n(1\rightarrow 0)+n(0\rightarrow 1)$, i.e., the number of households whose pass/fail classification differs across the paired models. $p$-values are from exact McNemar (binomial) tests based on the switchers. \raggedright
\end{table}

\vspace{-1em}

Taken together, these results indicate that while the precise \emph{allocation} of habits across characteristics plays little role in classification, introducing intertemporal dependence in characteristics space improves rationalisability for a limited but directionally important subset of households. In this application, the main empirical action comes from the maintained hedonic representation, while habits matter at the behavioural margin. The evidence is therefore not that cereal purchases reveal large sugar or sodium habits. It is that the framework separates a demanding price-representation problem from a behavioural question about whether static valuation is a useful approximation.

\enlargethispage{1\baselineskip}
\subsubsection{Structural and behavioural sources of empirical discipline}
\label{sec:restrictiveness}

If the test is to serve as an interpretation check for recovered WTP objects, rejection cannot be the end of the analysis. A binary pass/fail outcome says whether the maintained model holds exactly, but not whether it fails marginally or substantially. It also does not reveal whether the problem is the hedonic price representation or intertemporal behaviour. I therefore separate rationalisability into structural and behavioural components and quantify violations on each margin using distance measures. The construction exploits the fact that representation theorems induce measures of rationality violations through the distance to their defining restrictions (e.g., \citet{andrews2026revealed}).

In the hedonic setting, this distinction admits a natural decomposition. The two possible sources of failure map directly into the model's restrictions. Rationalisability in the habits-over-characteristics model requires joint satisfaction of two conceptually distinct restrictions. 
\emph{Structural equalities} (B2$^{+}$) link observed prices to product characteristics through the hedonic technology and act as overidentifying restrictions on the admissible shadow-price representation. This is the gatekeeper stage of the empirical analysis.
\emph{Behavioural inequalities} (B1$^{+}$) constrain intertemporal choice through shadow prices and the discount factor (Theorem~\ref{Afriat_Missing_Prices}). Separating these margins clarifies where empirical discipline originates in characteristics-based valuation and how it differs from more flexible goods-based representations.

The structural equalities require that, in each household-period, the observed price vector $\bm{\rho}_{t}^{+}$ lies in the column space of the augmented characteristics matrix $\Btil_t := \begin{bmatrix} \bm{B}_t^{\prime} \mid (\bm{B}_t^{a})^{\prime} \end{bmatrix}$. In principle this is a dimensionality-reduction restriction: prices must be representable through a lower-dimensional characteristics technology rather than arbitrary goods-specific shifters. In the scanner environment studied here, however, the active choice sets are typically narrow, so the empirical bite of the restriction depends on the realised rank of the purchased bundle. When the purchased-good matrix has full row rank, the column space of $\Btil_t$ spans all of $\mathbb{R}^{K_t^+}$ and the structural equalities are mechanically satisfied. Non-zero structural distances therefore arise precisely in those household-periods where the active bundle is rank deficient, either because the household buys too few linearly independent products or because the purchased UPCs have highly similar characteristic profiles.

This is empirically plausible in cereal data, where closely related UPCs often differ only in package size, branding, or minor formulation details and therefore carry very similar measured characteristic vectors. Even after conditioning on a rich set of nutritional and descriptive characteristics, scanner prices also reflect retailer pricing strategies, temporary promotions, and mark-ups driven by market power or shelf placement that do not correspond to attributes households consume. When the active bundle is low rank, those retailer-specific price components cannot be absorbed by the hedonic system and therefore appear as structural residuals.

To quantify the severity of these violations, I compute for each household-period the Euclidean distance
\[
d_t = \big\| \bm{\rho}_{t}^{+} - \Btil_t \Btil_t^{+} \bm{\rho}_{t}^{+} \big\|,
\]
where $\Btil_t^{+}$ denotes the Moore--Penrose pseudoinverse and $\Btil_t \Btil_t^{+}$ is the orthogonal projector onto the equality manifold implied by the hedonic representation. Equivalently, observed prices admit the orthogonal decomposition
\[
\bm{\rho}_t^{+} \;=\; \Btil_t \hat{\bm{\pi}}_t \;+\; \bm{r}_t,
\]
where $\hat{\bm{\pi}}_t$ minimises $\|\bm{\rho}_t^{+} - \Btil_t \bm{\pi}\|$ and $\bm{r}_t$ is orthogonal to the column space of $\Btil_t$. The distance $d_t = \|\bm{r}_t\|$, measured in dollars, therefore captures the minimal joint price adjustment required for a hedonic price representation to exist. For example, $d_t = 1.0$ means that the minimum adjustment vector to prices has Euclidean norm 1 dollar; equivalently, the sum of squared price adjustments across the goods purchased that period is 1.

One response to violations of the hedonic equalities is to augment the technology with latent characteristics, as in \citet{blow_revealed_2008}, thereby expanding the price manifold until the equalities hold. In the limit, sufficiently rich product-specific latent attributes collapse the characteristics model toward a goods-space representation, mechanically reducing structural violations. I do not pursue this route here. My objective is not to restore rationalisability by construction, but to ask how much discipline is imposed by the observed, policy-relevant characteristics and to separate structural misspecification from behavioural inconsistency.

Figure~\ref{fig:distance_hist} plots the distribution of mean distances across households for both characteristics- and goods-based specifications. The characteristics models exhibit a wide, right-skewed distribution centred above zero. In this application, those distances should be read as evidence that many active bundles fail to span their own purchased-good price space once prices are projected through the maintained characteristics technology. Even so, the implied unit-price violations are often modest in magnitude: most household-periods require a minimum joint price adjustment with Euclidean norm below 1 dollar to restore hedonic consistency. As a rough benchmark, average household cereal expenditure is \$22.1 per period. By contrast, goods-based models impose no cross-good price restrictions and therefore satisfy the structural equalities mechanically, yielding $d_t = 0$.

\vspace{-0.5em}
\begin{figure}[H]
\centering
\caption{Mean distance-to-manifold}
\vspace{-1em}
\label{fig:distance_hist}
\includegraphics[width=0.85\textwidth]{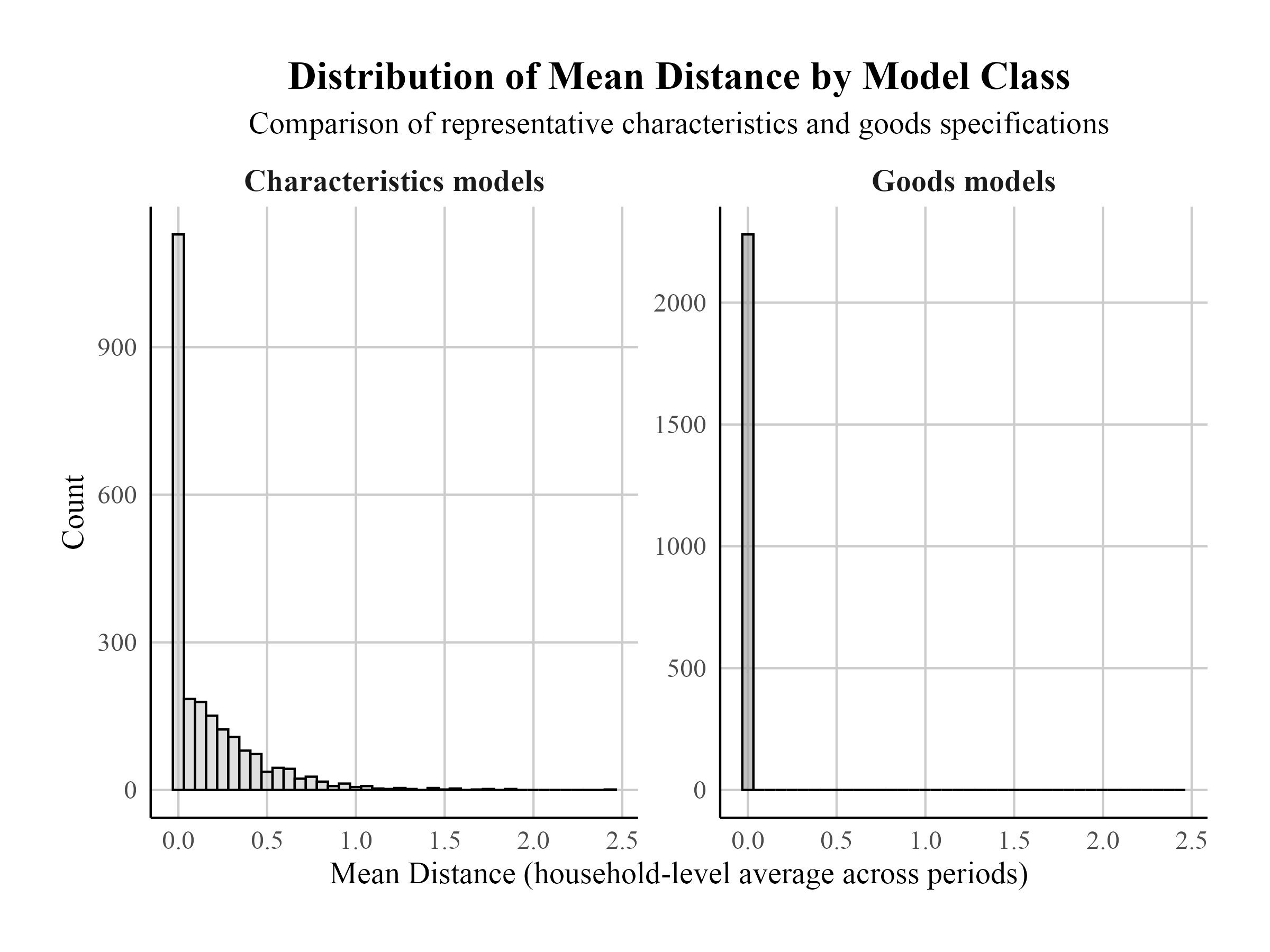}
\footnotesize
\begin{flushleft}
\textit{Notes:} Histograms of the household-level mean of the structural distance measure across periods. 
Distances are measured in dollars and represent the minimal joint price adjustment required for the hedonic price equalities to hold. 
Left panel reports the characteristics specification; right panel the goods specification.
\end{flushleft}
\end{figure}

\vspace{-2em}

An important implication of this geometry is that all characteristics-based specifications impose identical structural restrictions. Habit-forming characteristics enter $\Btil_t$ only as duplicated rows of $\bm{B}_t$ and therefore neither increase its rank nor enlarge the feasible price set. Structural restrictiveness is thus governed entirely by the hedonic representation itself, not by the allocation of habits across characteristics. In the present data, this means that structural rejection is best understood as a statement about the linear algebra of narrow purchased bundles---their rank and characteristic collinearity---rather than about the dynamic specification per se. Moreover, these structural restrictions operate at the level of prices and characteristics and are conceptually distinct from inventory dynamics: even if households smooth consumption through inventories, prices must still admit a hedonic representation for shadow prices to be well defined.

To measure behavioural restrictiveness, I adapt the Critical Cost Efficiency Index (CCEI) following \citet{afriat1973system} and \citet{varian_goodness--fit_1990}. The CCEI measures the smallest proportional relaxation of revealed affordability required for the behavioural inequalities in (B1$^{+}$) to admit a solution.\footnote{The CCEI is interpreted here strictly as a measure of RP slack, following its original cost-efficiency interpretation in \cite{afriat1973system}. As emphasized by \citet{echenique_meaning_2022}, it should \textit{not} be interpreted as a welfare loss or a measure of foregone surplus.} Economically, a CCEI of $\eta$ indicates that the revealed-affordability comparisons in the lifecycle RP test need only be relaxed by at most $(1-\eta)\times 100\%$ for the observed intertemporal choice path to become rationalisable. Values close to one therefore indicate that behaviour is nearly dynamically rational, while lower values signal more severe violations. The index is defined only for households whose prices lie exactly on the equality manifold, because characteristic shadow prices---and hence the behavioural inequalities themselves---are well defined only when the hedonic equalities hold. In this sense, the CCEI plays an analogous role to a goodness-of-fit statistic for intertemporal choice: it quantifies how much slack must be introduced before the model can rationalise observed behaviour, holding the price system fixed.

Figure~\ref{fig:CCEI_hist} shows that behavioural violations are generally modest, but systematically smaller when habits are allowed. All characteristics-based specifications exhibit substantial mass near unity, indicating that once the hedonic structure is satisfied, only limited perturbations are needed to rationalise behaviour. However, the static characteristics model without habits displays a thicker lower tail, reflecting greater intertemporal inconsistency even when structural feasibility holds. A parallel pattern arises in the goods domain: the habits-over-goods specification yields CCEIs tightly concentrated near one, while the static goods model exhibits a wider distribution with more mass below 0.95. These patterns indicate that allowing for habits reduces measured behavioural violations by absorbing intertemporal reversals that the static model cannot rationalise.

This behavioural improvement matters for interpretation. For households that satisfy the dynamic characteristics model but fail the static counterpart, the data reject the restriction that characteristic shadow prices can be read purely as contemporaneous marginal valuations. In those cases, a valuation exercise that treats sugar or sodium prices as static WTP imposes an additional timing restriction not supported by the observed purchase path. The empirical test therefore identifies when the dynamic wedge characterised in Section~\ref{sec:valuation_wedge} is supported.

The above distance measures show where and how severely the model fails, but they do not by themselves show whether the restrictions are demanding in this scanner environment. This is the fit-versus-restrictiveness concern emphasised by \citet{selten_properties_1991}: a model may rationalise many datasets because it captures structure, or because nearby alternatives would also pass. I therefore complement the decomposition with a simulation benchmark in the spirit of \citet{fudenberg_how_2023}, which evaluates how unusually well the model fits the observed data relative to nearby feasible perturbations of the scanner environment. The benchmark preserves each household's zero pattern and total expenditure while allowing prices and quantities to vary locally, thereby inducing a comparison distribution over structural and behavioural discrepancy measures. Full details are provided in Supplemental Appendix~\ref{sec:power_analysis_appendix}.

\vspace{-0.5em}
\begin{figure}[H] 
\centering 
\caption{Critical Cost Efficiency Index (CCEI)}
\label{fig:CCEI_hist} 
\includegraphics[width=\textwidth]{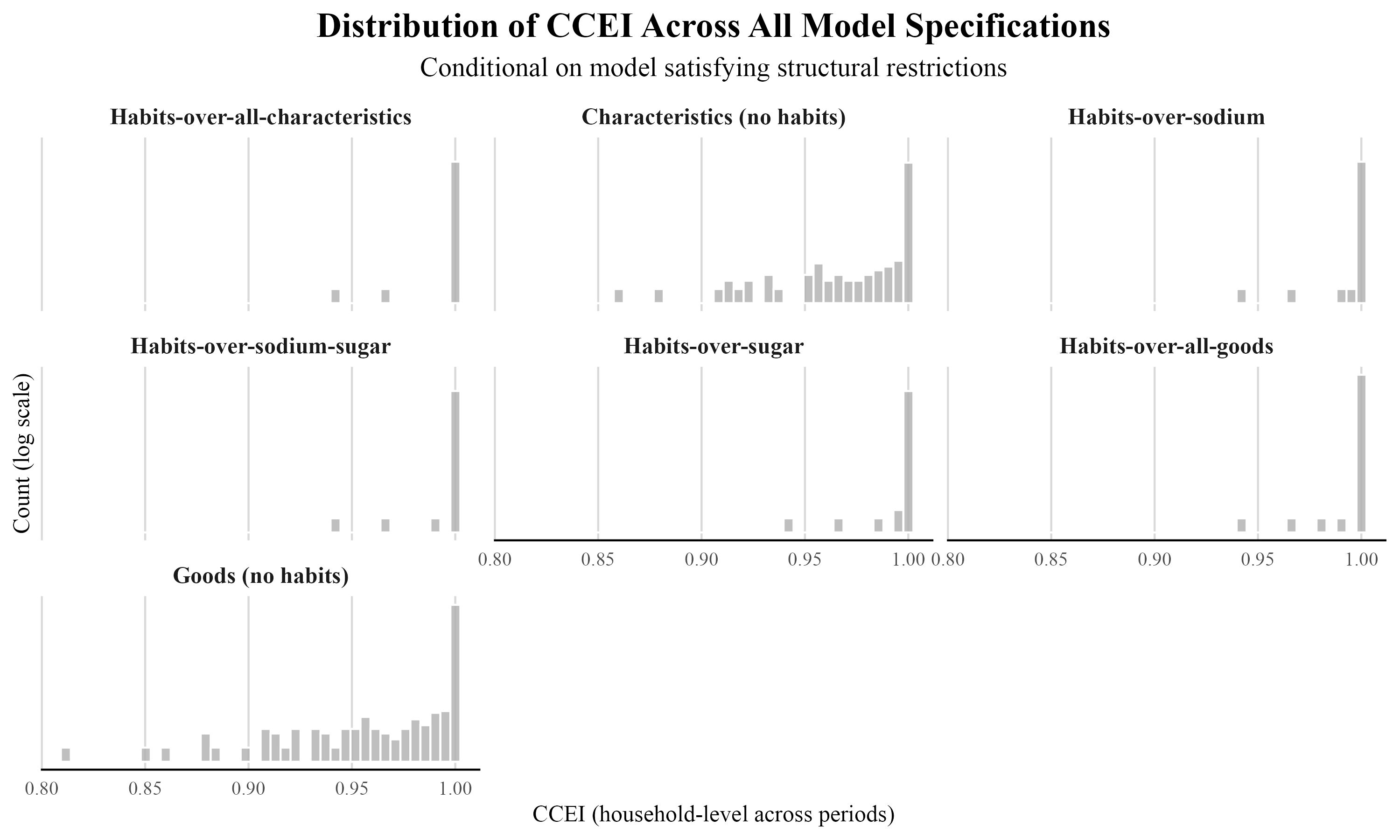} 
\footnotesize
\begin{flushleft}
\textit{Notes:} Histograms of household-level CCEI values across model specifications. 
The CCEI is the smallest proportional relaxation of the behavioural RP inequalities required for feasibility; values closer to one indicate smaller violations. 
Computed only for households satisfying the structural price equalities. 
Vertical axis is on a log scale.
\end{flushleft}
\end{figure} 

\vspace{-1.5em}

On the structural margin, the benchmark confirms that the hedonic equalities impose substantial empirical discipline. Observed scanner prices are, on average, closer to the equality manifold than locally perturbed price systems, yet are rarely extreme outliers relative to the induced comparison distribution. This indicates that the hedonic restrictions are demanding and not trivially satisfied. By contrast, goods-based models impose no cross-good price restrictions and therefore cannot fail on the structural margin, accounting for their much higher raw pass rates.

On the behavioural margin, CCEI values are extremely close to one in both the observed and locally perturbed data, leaving little scope for behaviour to appear unusually efficient in a quantile sense. This reflects the empirical environment---sparse active choice sets and limited intertemporal budget variation---rather than a lack of behavioural content in the model.

Taken together, the price geometry, behavioural slack measures, and simulation benchmark clarify the interpretation of raw pass rates. Structural restrictions embedded in the hedonic representation largely determine differences across goods and characteristics models, but in this application they do so through the rank properties of the active purchased bundles rather than through a broad market-level dimension count alone. Conditional on those restrictions, allowing for habit formation systematically improves intertemporal coherence relative to static preferences. Goods-based models achieve high pass rates primarily because they impose little structural content, not because they deliver a tighter account of behaviour.

\vspace{-0.5em}
\subsubsection{Discount factors}

The RP conditions are evaluated conditional on the discount factor $\beta$. Supplemental Appendix~\ref{app:beta} reports the fraction of rationalisable households consistent with each value on a fine grid $\beta \in [0.95,1]$. Acceptance probabilities are uniformly high across the grid under both habits-over-characteristics and habits-over-goods specifications, with no economically meaningful monotonic pattern.

These results indicate that, conditional on the hedonic price restrictions, the behavioural inequalities impose only weak discipline on intertemporal discounting in this environment. The identification sets for $\beta$ are wide, cautioning against interpreting discount factors recovered from nonparametric dynamic RP tests as tightly identified structural parameters in scanner data.

\vspace{-0.5em}
\subsubsection{Predictors of rationalisability}

Pass/fail outcomes are primarily associated with the scale and complexity of the observed choice problem rather than with demographics per se. Table~\ref{tab:probit_marginals_combined} reports average marginal effects from probit regressions of the pass indicator on household characteristics. Column (1) includes demographics only; Column (2) adds measures of purchasing intensity and variety.

In the demographic-only specification, households with children are 10.6 percentage points less likely to pass. Once purchasing controls are added, this effect attenuates and becomes statistically insignificant. By contrast, scale measures remain economically large and significant: an interquartile increase in two-year cereal expenditure is associated with roughly a 25 percentage point lower probability of passing, while a comparable increase in product variety reduces the probability by about 8 percentage points.

Taken together, these results suggest that household composition matters primarily through the scale and complexity of the observed choice problem. Larger and more diverse purchasing patterns generate a greater number of structural and behavioural constraints, increasing the likelihood that at least one is violated. For instance, when a household purchases products with sharply different nutritional profiles across periods, the implied cross-attribute shadow-price system must rationalise a wider range of price--quantity trade-offs than in households that repeatedly buy a narrow set of similar cereals. Demographics per se have limited explanatory power once this effective dimensionality is accounted for. The demographics-only specification has almost no explanatory power (McFadden pseudo-$R^{2}=0.009$), whereas the specification with purchasing controls reaches a more substantial 0.176.

\FloatBarrier
\begin{table}[H]\centering 
\caption{Determinants of rationalisability: probit average marginal effects}
  \label{tab:probit_marginals_combined} 
  \vspace{-0.2em}
\begin{tabular}{@{\extracolsep{5pt}}lD{.}{.}{-3}D{.}{.}{-3}} 
\\[-1.8ex]\hline 
\hline \\[-1.8ex] 
 & \multicolumn{2}{c}{\textit{Dependent variable: Pr(Pass)}} \\ 
\cline{2-3}
\\[-1.8ex] 
 & \multicolumn{1}{c}{(1) demographics only} 
 & \multicolumn{1}{c}{(2) w/ purchasing controls} \\ 
\hline \\[-1.8ex] 

young child 
& -0.106 
& -0.081 \\ 
& (0.084) 
& (0.079) \\[0.5ex]

children 
& -0.106^{**} 
& -0.012 \\ 
& (0.033) 
& (0.030) \\[0.5ex]

age group HH head 
& 0.004 
& -0.006 \\ 
& (0.012) 
& (0.011) \\[0.5ex]

education HH head 
& -0.00004 
& 0.001 \\ 
& (0.008) 
& (0.007) \\[0.5ex]

Pittsfield 
& 0.005 
& -0.001 \\ 
& (0.021) 
& (0.022) \\[0.5ex]

high income HH$^\dagger$ 
& -0.036 
& 0.008 \\ 
& (0.025) 
& (0.023) \\[0.5ex]

total units purchased 
&  \multicolumn{1}{c}{} 
& 0.00002 \\ 
&  
& (0.001) \\[0.5ex]

total expenditure 
&  
& -0.002^{***} \\ 
&  
& (0.0003) \\[0.5ex]

total unique products 
&  
& -0.003^{**} \\ 
&  
& (0.001) \\[0.5ex]

\hline \\[-1.8ex] 
Pseudo-$R^2$ (McFadden) 
& 0.009 
& 0.176 \\ 
Observations 
& \multicolumn{1}{c}{2,187} 
& \multicolumn{1}{c}{2,187} \\ 
\hline 
\hline \\[-1.8ex] 
\multicolumn{3}{l}{\footnotesize $^{*}$p$<$0.1; $^{**}$p$<$0.05; $^{***}$p$<$0.01.} \\
\multicolumn{3}{l}{\footnotesize $^\dagger$High-income indicator equals one for pre-tax income $\ge 75$k.} \\
\end{tabular} 
\vspace{-0.3em}
\footnotesize
\begin{flushleft}
\textit{Notes:} Entries report average marginal effects from probit models of an indicator for passing the habits-over-characteristics test. Column (1) includes demographic controls only; column (2) additionally includes purchasing intensity measures. Standard errors in parentheses are based on the delta method. 95 observations omitted due to missing demographic data.
\end{flushleft}
\end{table}
\FloatBarrier

\vspace{-2em}

\vspace{-1.7em}
\section{Conclusion}
\vspace{-0.5em}

This paper provides a nonparametric foundation for valuing characteristics when those characteristics may be habit forming. Policies often target attributes such as sugar, sodium, alcohol, caffeine, or emissions, and applied work often interprets characteristic prices or demand coefficients as static WTP. I show that a purely contemporaneous interpretation is valid only under a restriction: current consumption of the characteristic must carry no continuation value.

The main result is an Afriat-type characterisation of dynamic hedonic valuation. Observed prices and choices admit a dynamic characteristics-based interpretation if and only if there exist characteristic shadow prices and a discount factor satisfying structural price equalities and dynamic behavioural inequalities. The structural equalities require observed goods prices to be representable through the maintained goods-to-characteristics technology. The behavioural inequalities require choices to be dynamically rationalisable at the implied shadow prices.

\enlargethispage{0.75\baselineskip}
The framework clarifies why dimensionality reduction and dynamics play different roles. Moving from goods to characteristics yields parsimony but imposes geometric discipline on the admissible representation of prices. Accommodating dynamic preferences does not relax that discipline. Rather, it changes the behavioural restrictions imposed once characteristic shadow prices exist: choices that violate static separability may still be rationalisable when current characteristics are allowed to affect future marginal valuations. This distinction matters for policy because static WTP may be informative at a fixed state, while static and dynamic counterfactuals can diverge when policy effects propagate through habit stocks.

The cereal application illustrates this logic in a purchase-based scanner setting. Goods-based benchmarks pass at very high rates, while characteristics-based models fail more often because observed prices frequently violate the hedonic spanning restriction. The associated distance measures show, however, that many structural violations are economically modest. Conditional on structural admissibility, behaviour is close to dynamically rational, and allowing for habit formation improves behavioural coherence relative to static characteristics models for a subset of households.

The framework is best understood as a diagnostic of admissibility and interpretation rather than as a point estimator of welfare primitives. Rationalisability does not identify a unique dynamic preference representation, and the continuation component of characteristic prices is generally set identified without further structure. This limitation is also informative: it shows exactly which additional assumptions are required before static WTP, dynamic WTP, or policy counterfactuals can be given a unique welfare interpretation. Empirical failure should be read in the same spirit: it points to an incomplete maintained representation rather than literal irrationality.

Several directions for future work follow directly. Richer data with denser price support, more frequent observation, and greater intertemporal budget variation would sharpen both structural and behavioural restrictions. On the modelling side, extensions to stochastic choice, richer unobserved heterogeneity, and environments in which utility depends on current goods while habits attach to lagged characteristics would broaden the scope of the test. The broader message is that WTP for product characteristics, whether recovered from prices or demand systems, is not automatically a sufficient statistic for policy counterfactuals: this paper provides a disciplined framework for testing whether characteristic shadow values exist, whether they can be read statically, and when their use in counterfactuals should account for continuation effects.

\newpage
\renewcommand{\appendixtocname}{Appendix}  
\renewcommand{\appendixpagename}{Appendix}  
\begin{appendices}
\addtocontents{toc}{\protect\setcounter{tocdepth}{-1}} 

\vspace{-1em}

\section{Proofs}
\label{sec:proofs}

\noindent
\textbf{Proof of Lemma \ref{consistency_conditions}.}
\label{sec:consumer_maximisation_FOC}

Substituting the technology constraint $\ztil_t = \Atil \xtil_t$ into the consumer's lifetime problem in \eqref{max_prob} and using the lifetime budget constraint to substitute out the outside good, the relevant Lagrangian can be written as
\begin{align}
\label{lagrangian}
\mathcal{L}(\{\bm{x}_t\})
= \sum_{t=1}^{T} \beta^{t-1}u(\Atil\xtil_t)
- \left\{\sum_{t=1}^{T} \bm{\rho}_t' \bm{x}_t - W\right\},
\end{align}
where $\bm{x}_0$ is treated as fixed, $W$ is interpreted as lifetime wealth net of outside-good consumption, and I normalise the multiplier on lifetime wealth to one.

Using the chain rule \citep{felippa_introduction_2004}, the relevant derivatives are $\partial u(\Atil\xtil_t)/\partial \bm{x}_t = [\bm{I}_{K \times K} \mid \bm{0}_{K \times K}]\Atil^{\prime}\partial u(\ztil_t)$ and $\partial u(\Atil\xtil_{t+1})/\partial \bm{x}_t = [\bm{0}_{K \times K} \mid \bm{I}_{K \times K}]\Atil^{\prime}\partial u(\ztil_{t+1})$.\footnote{Since $u(\cdot)$ is assumed to be concave but not necessarily differentiable, each $\partial u(\ztil_t)$ denotes a supergradient. As such, expressions like those below should formally be interpreted as \textit{set inclusions} rather than strict equalities---e.g., $0 \in \partial_{\bm{x}_t} \mathcal{L}$. When $u$ is differentiable at $\ztil_t$, the supergradient is a singleton and the equality holds exactly.}
Here $\partial u(\ztil_t):=[\partial_{\bm{z}_t^c}u(\ztil_t)^{\prime}, \partial_{\bm{z}_t^a}u(\ztil_t)^{\prime}, \partial_{\bm{z}_{t-1}^a}u(\ztil_t)^{\prime}]^{\prime}$.

For interior dates $t \in \{1,\ldots,T-1\}$, the relevant first-order necessary conditions associated with the Lagrangian in
\eqref{lagrangian} now follow immediately as,
\begin{align}
\label{FOC_1}
    \partial _{\bm{x}_t} \mathcal{L} = 0 
    \, \Rightarrow \, 
    &
\bm{\rho}_t 
     =    
    \beta^{t-1}
    \left( 
    \begin{bmatrix}  \bm{I}_{K \times K} \, \big\rvert \, \bm{0}_{K \times K} \end{bmatrix}
    \Atil^{\prime}
    \partial u(\ztil_t) 
    +
    \beta \begin{bmatrix} \bm{0}_{K \times K} \, \big\rvert \,  \bm{I}_{K \times K}  \end{bmatrix} 
\Atil^{\prime}
\partial u(\ztil_{t+1})       
    \right) .
\end{align}
But recall from \eqref{augmented_notation} that $\Atil$ is a $2K \times (J+J_2)$ block matrix. 
Since the supergradient $\partial u(\ztil_t)$ can also be partitioned as a $J + J_2$ block vector,
$\partial u(\ztil_t)
:=
[
(
\partial_{\bm{z}_t^c} u(\ztil_t)
)^{\prime},
(
\partial_{\bm{z}_t^a} u(\ztil_t)
)^{\prime},
\partial_{\bm{z}_{t-1}^a} u(\ztil_t)^{\prime}
]^{\prime}$,
these FOCs reduce to:
\vspace{-1em}
\begin{align}
\bm{\rho}_t 
& =
\label{FOC_3}
\beta^{t-1} \left(
\bm{A}^{\prime} 
\begin{bmatrix}
\partial_{\bm{z}_t^c} u(\ztil_t)    \\
\partial_{\bm{z}_t^a} u(\ztil_t)   
\end{bmatrix} 
+
\beta   (\bm{A}^a)^{\prime}
\begin{bmatrix} 
\partial_{\bm{z}_{t}^a} u(\ztil_{t+1}) 
\end{bmatrix}      
\right).
\end{align}
At the terminal date, the continuation term vanishes. Hence the corresponding  FOC reduces to
\begin{align}
\label{FOC_terminal}
\bm{\rho}_T
&=
\beta^{T-1}
\bm{A}^{\prime}
\begin{bmatrix}
  \partial_{\bm{z}_T^c} u(\ztil_T)    \\
  \partial_{\bm{z}_T^a} u(\ztil_T)
\end{bmatrix},
\end{align}
which is exactly the same pricing condition as in \eqref{FOC_3} under the convention $\bm{\pi}_{T+1}^1 \equiv \bm{0}$.
Since I assume $u$ to be concave, the associated KKT conditions are sufficient
for a global maximum. Hence, allowing for corner solutions, I replace the
stationarity equalities in \eqref{FOC_3} and \eqref{FOC_terminal} with inequalities to obtain the full set
of price and data pairs $\{\bm{\rho}_t,\bm{x}_t\}_{t=1}^T$ consistent with an
interior or corner solution to the consumer's maximisation problem. This gives
rise to my formal definition of \textit{consistency} in Definition
\ref{consistency_definition} and Lemma \ref{consistency_conditions}. $\qed$

\noindent
\textbf{Proof of Theorem \ref{Afriat}.}

\noindent
$(A) \Rightarrow(B)$: Assume $(A)$ holds. Then, by Lemma \ref{consistency_conditions}, $\bm{\rho}_t \geq \bm{A}^{\prime} \bm{\pi}_t^{0}+(\bm{A}^a)^{\prime} \bm{\pi}_{t+1}^{1}$ for all $t \in \{1, \ldots, T\}$, with equality for all $k$ such that $x_t^k > 0$, where $\bm{\pi}_{T+1}^1 \equiv \bm{0}$.
Element wise, this states that for all $k$ and for all $t \in \{1, \ldots, T\}$,
$
\rho^k_t
\geq 
\bm{a}_k^{\prime} \bm{\pi}_t^0 +
\bm{a}_k^{a \prime}  \bm{\pi}_{t+1}^{1},
$
and where $x_t^k > 0$,
$
\rho^k_t
= 
\bm{a}_k^{\prime} \bm{\pi}_t^0 +
\bm{a}_k^{a\prime}  \bm{\pi}_{t+1}^{1}.
$
This gives us restrictions \eqref{(B2)} and \eqref{(B3)}, respectively.

It remains to show \eqref{(B1)} holds. Using the augmented notation
$
\ztil_t := ((\bm{z}_t^c)^{\prime}, (\bm{z}_t^a)^{\prime}, (\bm{z}_{t-1}^a)^{\prime})^{\prime}
$
and
$
\pitil_t := \frac{1}{\beta^{t-1}}\left[\bm{\pi}_t^{0 \prime}, \bm{\pi}_t^{1 \prime}\right]^{\prime}
$
it follows from the shadow-price definitions in \eqref{shadow_price_0} that for any element of the superdifferential of $u$ at $\ztil_t$,
\vspace{-1em}
\begin{align*}
\partial u(\ztil_t)^{\prime}
& = \left[ \left[\partial_{\bm{z}^c_t}u(\ztil_t)^{\prime}, \partial_{\bm{z}^a_t}u(\ztil_t)^{\prime}\right], \partial_{\bm{z}^a_{t-1}}u(\ztil_t)^{\prime}  \right] 
= \frac{1}{\beta^{t-1}}\left[\bm{\pi}_t^{0 \prime}, \bm{\pi}_t^{1 \prime}\right]
= \pitil_t^{\prime}.
\end{align*}
\vspace{-0.5em}
The concavity and superdifferentiability of $u(\ztil_t)$ imply $u(\ztil_s)-u(\ztil_t) \leq \partial u(\ztil_t)^{\prime}(\ztil_s-\ztil_t)$ for all $s,t \in\{1, \ldots, T\}$. Combining this concavity result with the optimising behaviour above gives $u(\ztil_s)-u(\ztil_t) \leq \pitil_t^{\prime}(\ztil_s-\ztil_t)$ for all $s,t \in\{1, \ldots, T\}$.
Take any finite ordered cycle $(t_1,\ldots,t_M)$ from $\tau$, with $t_{M+1}=t_1$. Summing the corresponding inequalities
$
u(\ztil_{t_{m+1}})-u(\ztil_{t_m})
\leq
\pitil_{t_m}^{\prime}(\ztil_{t_{m+1}}-\ztil_{t_m}),$ for $m=1,\ldots,M
$
gives $0 \leq \sum_{m=1}^{M}\pitil_{t_m}^{\prime}(\ztil_{t_{m+1}}-\ztil_{t_m})$, which is restriction \eqref{(B1)}.

\noindent
$(B) \Rightarrow(A)$:
Restriction \eqref{(B1)} imposes that the finite dataset $\left(\pitil_t, \ztil_t\right), t = 1, \ldots, T$ is cyclically monotone as defined by \cite{browning_nonparametric_1989}. By the standard finite-sample Afriat inequalities \citep{afriat_construction_1967,diewert_afriat_1973,varian_nonparametric_1982}, this is equivalent to the existence of $T$ numbers $\{V_t\}_{t=1}^T$ such that $V_s \leq V_t + \pitil_t^{\prime}(\ztil_s-\ztil_t)$ for all $s,t \in \tau$.
Applying Afriat's theorem to the finite dataset $\left\{ \pitil_t, \ztil_t \right\}_{t \in \tau}$ then implies that there exists a locally non-satiated and concave utility function rationalising these observations. One such representation is the Afriat envelope,
\vspace{-1em}
\begin{align*}
u(\ztil)
:=
\min_{t \in \tau}
\left\{
V_t + \pitil_t^{\prime}(\ztil-\ztil_t)
\right\}.
\end{align*}
\vspace{-0.5em}
Since $u(\cdot)$ is the minimum of finitely many affine functions, it is proper, concave, and superdifferentiable. Moreover, the Afriat inequalities imply that at each observed $\ztil_t$ the $t$-th affine function supports $u$, so that
$
\pitil_t \in \partial u(\ztil_t), \, \forall \, t \in \tau.
$
Selecting one such supergradient at each observed point and partitioning it conformably yields the following relationships between shadow prices and supergradients of the utility function:
\begin{align}
\label{cyc_mon_1}
& \bm{\pi}_t^{0} = 
\beta^{t-1}
\begin{bmatrix}
  \partial_{\bm{z}_t^c} u(\ztil_t)    \\
  \partial_{\bm{z}_t^a} u(\ztil_t)   
\end{bmatrix}  
\qquad \text{and} \qquad
\bm{\pi}_{t}^{1} 
= 
\beta^{t-1} 
\begin{bmatrix} 
\partial_{\bm{z}_{t-1}^a} u(\ztil_{t}) 
\end{bmatrix}.
\end{align}
Since these hold for all $t \in \tau$, the latter condition can be forwarded one period to obtain,
\begin{align}
\label{cyc_mon_3}
& 
\bm{\pi}_{t+1}^{1}
=
\beta^{t}
\begin{bmatrix} 
\partial_{\bm{z}_{t}^a} u(\ztil_{t+1}) 
\end{bmatrix}
\end{align}
for all $t \in \{1, \ldots, T-1\}$. 
Substituting expressions \eqref{cyc_mon_1} and \eqref{cyc_mon_3} into condition \eqref{(B2)} re-written in matrix form,
$
\bm{\rho}_t 
\geq 
\bm{A}^{\prime} \bm{\pi}_{t}^0 + 
(\bm{A}^a)^{\prime}  \bm{\pi}_{t+1}^{1},
$
I obtain,
\begin{align}
\label{consistency_proof}
\bm{\rho}_t 
& \geq 
\bm{A}^{\prime}
\beta^{t-1}
\begin{bmatrix}
  \partial_{\bm{z}_t^c} u(\ztil_t)    \\
  \partial_{\bm{z}_t^a} u(\ztil_t)   
\end{bmatrix} 
+ 
(\bm{A}^a)^{\prime} 
\beta^{t}
\begin{bmatrix} 
\partial_{\bm{z}_{t}^a} u(\ztil_{t+1}) 
\end{bmatrix}.
\end{align}
At the terminal date, condition \eqref{(B2)} reduces under $\bm{\pi}_{T+1}^1 \equiv \bm{0}$ to
$
\bm{\rho}_T
\geq
\bm{A}^{\prime}
\beta^{T-1}
\begin{bmatrix}
  \partial_{\bm{z}_T^c} u(\ztil_T)^{\prime}    \quad
  \partial_{\bm{z}_T^a} u(\ztil_T)^{\prime}
\end{bmatrix}^{\prime}.
$
In the event that a good $k$ is consumed in strictly positive amounts, then substituting expressions \eqref{cyc_mon_1} and \eqref{cyc_mon_3} into condition \eqref{(B3)} yields the same expression as in \eqref{consistency_proof}, but with equality; the terminal-date analogue again follows by setting $\bm{\pi}_{T+1}^1 \equiv \bm{0}$.
By Lemma \ref{consistency_conditions}, this means that the data $\left\{\bm{\rho}_t; \bm{x}_t\right\}_{t \in \{1, \ldots, T\}}$ are consistent with the one-lag habits model for given technology $\bm{A}$. $\qed$

\noindent
\textbf{Proof of Theorem \ref{Afriat_Missing_Prices}.}
\label{sec:missing_prices_proof}

\noindent
$(A^+) \Rightarrow(B^+)$ : Given $\{\bm{\rho}^0_t\}_{t \in\{1, \ldots, T\}}$ I have a full set of prices, $\{\bm{\rho}_t\}_{t \in\{1, \ldots, T\}}$, such that the data satisfies the model. Hence, by Theorem \ref{Afriat}, condition $(B)$ holds. Hence, $(B^+)$ also holds.

\noindent
$(B^+) \Rightarrow(A^+)$ : I have shadow discounted prices $\left\{\bm{\pi}_t^{0}\right\}_{t \in\{1, \ldots, T\}}$ and $\left\{\bm{\pi}_t^{1}\right\}_{t \in\{1, \ldots, T\}}$ such that \eqref{(B1_missing)} and \eqref{(B2_missing)} hold. Use these shadow prices to construct the unobserved prices as $\bm{\rho}^0_t =(\bm{B}^0_t)^{\prime} \bm{\pi}_t^0 +(\bm{B}^{a, 0}_t)^{\prime}  \bm{\pi}_{t+1}^{1}$, with $\bm{\pi}_{T+1}^1 \equiv \bm{0}$. Using these constructed prices, it follows by comparison with Theorem \ref{Afriat} that (A$^+$) holds. $\qed$


\noindent
\textbf{Habits over goods as a special case of the characteristics model (Definition \ref{consistency_goods}).}

The (one-lag) habits-over-characteristics model nests the habits-over-goods model of \cite{crawford_habits_2010} as a special case. To see this, consider a trivial characteristics model where $J = K$ and the technology matrix $\bm{A}$ is the $J \times J$ identity matrix. Then:
\[
\bm{z}_t = 
\begin{bmatrix}
\bm{z}_t^c \\
\bm{z}_t^a
\end{bmatrix}
= 
\begin{bmatrix}
\bm{x}_t^c \\
\bm{x}_t^a
\end{bmatrix}
= \bm{x}_t,
\]
where $\bm{x}_t^c$ and $\bm{x}_t^a$ denote the $J_1$ non-habit-forming and $J_2$ habit-forming goods, respectively. Under this assumption, the matrices $(\bm{A}^c)^{\prime}$ and $(\bm{A}^a)^{\prime}$ take on block-structured forms. The matrix $(\bm{A}^c)^{\prime}$ is a $K \times J_1$ matrix whose first $J_1$ rows form a $J_1 \times J_1$ identity matrix, with all remaining entries equal to zero. Conversely, $(\bm{A}^a)^{\prime}$ is a $K \times J_2$ matrix whose bottom $J_2$ rows comprise a $J_2 \times J_2$ identity matrix, while the entries in the first $K - J_2$ rows are zero.

Substituting the identity structure of $\bm{A}$ into Equation \eqref{foc_market_prices} and defining the augmented bundle $\bar{\bm{x}}_t := \left( \bm{x}_t^{c \prime}, \bm{x}_t^{a \prime}, \bm{x}_{t-1}^{a \prime} \right)^{\prime} = \ztil_t$, the FOC becomes,
\begin{align*}
\bm{\rho}_t 
&  =
\begin{bmatrix}
\bm{\rho}_t^c \\
\bm{\rho}_t^a
\end{bmatrix}
 \geq
\begin{bmatrix}
\bm{\pi}_{t}^{c, 0} \\
\bm{0}_{K - J_1}
\end{bmatrix} 
+
\begin{bmatrix}
\bm{0}_{K - J_2} \\
\bm{\pi}_{t}^{a, 0}
\end{bmatrix} 
+
\begin{bmatrix}
\bm{0}_{K - J_2} \\
\bm{\pi}_{t+1}^{1}
\end{bmatrix} 
=
\beta^{t-1} 
\begin{bmatrix}
\partial_{\bm{x}_t^c} u(\bar{\bm{x}}_t) \\
\bm{0}_{K - J_1}
\end{bmatrix} 
+
\beta^{t-1}
\begin{bmatrix}
\bm{0}_{K - J_2} \\
\partial_{\bm{x}_t^a} u(\bar{\bm{x}}_t) 
\end{bmatrix} 
+
\beta^{t} 
\begin{bmatrix}
\bm{0}_{K - J_2} \\
\partial_{\bm{x}_t^a} u(\bar{\bm{x}}_{t+1}) 
\end{bmatrix},
\end{align*}
where $\bm{\pi}_{t}^{c, 0}$, $\bm{\pi}_{t}^{a, 0}$ and $\bm{\rho}_t^{c}$, $\bm{\rho}_t^{a}$ are the subvectors of $\bm{\pi}_t^0$ and $\bm{\rho}_t$ corresponding to non-habit-forming and habit-forming goods, respectively, and the final equality follows using the shadow-price definitions in \ref{shadow_price_0}.
For all $k$ such that $x_t^k > 0$, these inequalities hold with equality. Thus, under the restriction $J = K$ and $\bm{A} = \bm{I}_J$, Definition \ref{consistency_definition} simplifies as in Definition \ref{consistency_goods}.

\noindent
\textbf{Proof of Corollary \ref{AfriatGoods}.}
\label{proof_afriat_goods}

\noindent
Follows directly from Theorem \ref{Afriat} under the restriction $J = K$ and $\bm{A} = \bm{I}_J$. Equivalently, one may derive it from Definition \ref{consistency_goods} by replicating the proof of Theorem \ref{Afriat}. $\qed$

\noindent
\textbf{Intertemporally separable preferences over characteristics (Definition \ref{consistency_sep}).}

To formalise this nesting result, assume $\bm{z}_t = \bm{z}_t^c$ and embed the model in a lifecycle framework with a single intertemporal budget constraint and maintain the normalisation that the marginal utility of lifetime wealth equals one.
Then the shadow price in Equation \eqref{shadow_price_0} reduces to
$
\bm{\pi}_t = \partial_{\bm{z}_t} u(\bm{z}_t),
$
noting that $\beta$ no longer appears due to intertemporal separability.\footnote{This is equivalent to setting $\beta = 1$, which is without loss of generality in this setting.} Under these assumptions, Definition \ref{consistency_definition} reduces to that in Definition \ref{consistency_sep}.

\noindent
\textbf{Proof of Corollary \ref{AfriatSep}.}
\label{proof_afriat_sep}

\noindent
Corollary \ref{AfriatSep} is immediate from Theorem \ref{Afriat} upon setting $J_1 = J$ and $\beta = 1$. 

\noindent
I now show that this Corollary is equivalent to Theorem 1 of \cite{blow_revealed_2008} when the static characteristics model is embedded in a lifecycle framework. Equivalence between (A'') and the P'' condition in their theorem follows directly from Definition \ref{consistency_sep}. Equivalence of conditions \eqref{(B2sep)} and \eqref{(B3sep)} with (A2) and (A3) in \cite{blow_revealed_2008} is straightforward after aligning notation.

\noindent
The key step is to show that condition \eqref{(B1sep)} is equivalent to condition (A1) in \cite{blow_revealed_2008}, which states that there exist $T$ scalars $V_t$ and $T$ vectors $\bm{\pi}_t$ such that:
\vspace{-1em}
\begin{align}
\label{(A1_Blow)}
V_s 
& \leq V_t + \lambda_t \bm{\pi}_t^{\prime}(\bm{A} \bm{x}_s - \bm{A} \bm{x}_t) \quad \forall \, s, t. \tag{A1}
\end{align}
In \citet{blow_revealed_2008}, $\lambda_t$ denotes the (time-$t$) marginal utility of wealth. Under my single lifetime budget constraint I have $\lambda_t \equiv \lambda$, and I maintain the normalisation $\lambda = 1$ throughout.

To show that \eqref{(B1sep)} implies \eqref{(A1_Blow)}, note that \eqref{(B1sep)} imposes cyclical monotonicity on the finite dataset $\left\{ \bm{\pi}_t, \bm{z}_t \right\}_{t=1}^T$. By the standard finite-sample Afriat inequalities \citep{afriat_construction_1967,diewert_afriat_1973,varian_nonparametric_1982}, this is equivalent to the existence of $T$ numbers $\{V_t\}_{t=1}^T$ such that $V_s \leq V_t + \bm{\pi}_t^{\prime}(\bm{z}_s - \bm{z}_t)$ for all $s,t$. Since $\bm{z}_t = \bm{A}\bm{x}_t$ in the static characteristics model, this becomes $V_s \leq V_t + \bm{\pi}_t^{\prime}(\bm{A}\bm{x}_s - \bm{A}\bm{x}_t)$ for all $s,t$, which is \eqref{(A1_Blow)} under the normalisation $\lambda_t \equiv 1$.

\noindent
Conversely, suppose \eqref{(A1_Blow)} holds. 
Under $\lambda=1$, rearranging yields $0 \leq (V_t - V_s) + \bm{\pi}_t^{\prime}(\bm{A} \bm{x}_s - \bm{A} \bm{x}_t)$.
Take any finite ordered cycle $(t_1,\ldots,t_M)$ with $t_{M+1}=t_1$. Summing the corresponding inequalities over $m=1,\ldots,M$ makes the Afriat numbers telescope, yielding
$0 \leq \sum_{m=1}^{M}\bm{\pi}_{t_m}^{\prime}(\bm{A}\bm{x}_{t_{m+1}}-\bm{A}\bm{x}_{t_m})$, which is \eqref{(B1sep)}.

\noindent
I conclude that Corollary \ref{AfriatSep} provides an equivalent characterisation to \cite{blow_revealed_2008} when the static characteristics model is embedded in a lifecycle framework. This highlights the additional empirical restrictions imposed by intertemporal optimisation: a consumer consistent with intratemporal utility maximisation over characteristics must also allocate expenditure across periods to maximise lifetime utility. $\qed$

\vspace{-1em}
\section{Additional data description and summary statistics}
\label{app:data_appendix}

This appendix reports additional descriptive statistics for the scanner panel, including sample selection details, purchase intensity, brand concentration, price distributions, and characteristics distributions.

\vspace{-1em}
\subsection{Sample selection}

I begin by documenting the construction of the estimation sample from the raw 2010--2011 scanner panel. Table \ref{tab:sample_cuts} reports the sequential filtering steps and resulting sample sizes. The first four restrictions impose minimum activity and panel-length requirements required for feasibility of the revealed preference framework and are therefore purely mechanical. The final restriction---dropping households whose purchased UPCs lack complete characteristics data---is the only potentially selective cut and is evaluated separately using balance tests.

As shown in Table \ref{tab:sample_cuts}, most sample attrition occurs due to mechanical feasibility restrictions, while the potentially selective exclusion due to missing characteristics data affects a comparatively smaller subset of households.

\begin{table}[H]
\centering
\caption{\label{tab:sample_cuts}
Sample construction and filtering}
\fontsize{10}{12}\selectfont
\begin{tabular}{lccc}
\toprule
Step & Households & UPCs (K) & Purely  \\
 &  &  &  mechanical \\
\midrule
raw scanner (2010--11 combined) &
4697 & 1139 & Yes \\

drop HHs with only one observed purchase occasion &
4505 & 1136 & Yes \\

drop HHs with only one constructed time period ($T_i \leq 1$) &
3673 & 1131 & Yes \\

drop HHs with $T_i \leq 2$ (not meaningful for habit formation) &
2972 & 1121 & Yes \\

drop HHs who purchase UPCs with missing characteristics data &
2282 & 801 & No \\
\bottomrule
\end{tabular}
\footnotesize
\begin{flushleft}
\textit{Notes:} Rows report sequential sample restrictions applied to the 2010--11 scanner panel. $T_i$ denotes the number of household-specific constructed time periods. The restriction $T_i \le 2$ ensures at least two transitions for identification of one-lag habit formation. ``Purely mechanical'' indicates whether the restriction is driven by panel structure rather than missing product characteristics.
\end{flushleft}
\end{table}

\vspace{-2em}
\begin{table}[H]
\centering
\caption{\label{tab:balance_demo}
Demographic balance: excluded versus final sample}
\fontsize{10}{12}\selectfont
\begin{threeparttable}
\begin{tabular}[t]{lcccc}
\toprule
\textbf{Characteristic} &
\makecell[c]{\textbf{Overall}\\ N = 2,972} &
\makecell[c]{\textbf{Excluded}\\ N = 690} &
\makecell[c]{\textbf{Final Sample}\\ N = 2,282} &
\textbf{p-value}\\
\midrule

\cellcolor{gray!10}\textbf{Family size} & \cellcolor{gray!10}{} & \cellcolor{gray!10}{} & \cellcolor{gray!10}{} & \cellcolor{gray!10}{0.012}\\
\hspace{1em}1   & 550 (19\%)   & 118 (17\%) & 432 (19\%)   & \\
\cellcolor{gray!10}\hspace{1em}2
                & \cellcolor{gray!10}{1,362 (46\%)} & \cellcolor{gray!10}{293 (42\%)} & \cellcolor{gray!10}{1,069 (47\%)} & \cellcolor{gray!10}{}\\
\hspace{1em}3+  & 1,060 (36\%) & 279 (40\%) & 781 (34\%)   & \\

\cellcolor{gray!10}\textbf{Marital status} & \cellcolor{gray!10}{} & \cellcolor{gray!10}{} & \cellcolor{gray!10}{} & \cellcolor{gray!10}{$>0.9$}\\
\hspace{1em}Married     & 307 (10\%)  & 70 (10\%)  & 237 (10\%)  & \\
\cellcolor{gray!10}\hspace{1em}Not married
                        & \cellcolor{gray!10}{2,658 (90\%)} & \cellcolor{gray!10}{617 (90\%)} & \cellcolor{gray!10}{2,041 (90\%)} & \cellcolor{gray!10}{}\\

\textbf{Region} &  &  &  & $<0.001$\\
\cellcolor{gray!10}\hspace{1em}Eau Claire
                & \cellcolor{gray!10}{1,597 (54\%)} & \cellcolor{gray!10}{278 (40\%)} & \cellcolor{gray!10}{1,319 (58\%)} & \cellcolor{gray!10}{}\\
\hspace{1em}Pittsfield
                & 1,375 (46\%) & 412 (60\%) & 963 (42\%) & \\
\cellcolor{gray!10}\textbf{Children} & \cellcolor{gray!10}{} & \cellcolor{gray!10}{} & \cellcolor{gray!10}{} & \cellcolor{gray!10}{0.032}\\
\hspace{1em}Any children & 667 (22\%)   & 176 (26\%) & 491 (22\%) & \\
\cellcolor{gray!10}\hspace{1em}No children
                        & \cellcolor{gray!10}{2,305 (78\%)} & \cellcolor{gray!10}{514 (74\%)} & \cellcolor{gray!10}{1,791 (78\%)} & \cellcolor{gray!10}{}\\

\textbf{Occupation} &  &  &  & 0.9\\
\cellcolor{gray!10}\hspace{1em}Blue collar/service
                & \cellcolor{gray!10}{1,149 (39\%)} & \cellcolor{gray!10}{267 (39\%)} & \cellcolor{gray!10}{882 (39\%)} & \cellcolor{gray!10}{}\\
\hspace{1em}Office/sales (sales/clerical)
                & 454 (15\%) & 100 (14\%) & 354 (16\%) & \\
\cellcolor{gray!10}\hspace{1em}Other/unknown
                & \cellcolor{gray!10}{215 (7.2\%)} & \cellcolor{gray!10}{55 (8.0\%)} & \cellcolor{gray!10}{160 (7.0\%)} & \cellcolor{gray!10}{}\\
\hspace{1em}Retired
                & 254 (8.5\%) & 62 (9.0\%) & 192 (8.4\%) & \\
\cellcolor{gray!10}\hspace{1em}White collar (prof/manager)
                & \cellcolor{gray!10}{900 (30\%)} & \cellcolor{gray!10}{206 (30\%)} & \cellcolor{gray!10}{694 (30\%)} & \cellcolor{gray!10}{}\\

\textbf{Education} &  &  &  & 0.8\\
\addlinespace
\cellcolor{gray!10}\hspace{1em}$\le$ High school
                & \cellcolor{gray!10}{1,737 (61\%)} & \cellcolor{gray!10}{400 (60\%)} & \cellcolor{gray!10}{1,337 (61\%)} & \cellcolor{gray!10}{}\\
\hspace{1em}College+
                & 203 (7.1\%) & 46 (6.9\%) & 157 (7.2\%) & \\
\cellcolor{gray!10}\hspace{1em}Some college/tech
                & \cellcolor{gray!10}{911 (32\%)} & \cellcolor{gray!10}{219 (33\%)} & \cellcolor{gray!10}{692 (32\%)} & \cellcolor{gray!10}{}\\

\bottomrule
\end{tabular}
\footnotesize
\begin{flushleft}
\textit{Notes:} Counts are n (\%). The ``Excluded'' column reports households removed due to missing UPC characteristic information; all other panel-length restrictions are satisfied in both groups. p-values are from Pearson chi-squared tests of equality across excluded and retained households.
\end{flushleft}
\end{threeparttable}
\end{table}

\vspace{-1em}
To assess whether exclusion due to missing characteristics data induces observable selection, Table \ref{tab:balance_demo} compares demographic characteristics between the baseline sample---defined as households satisfying all mechanical feasibility restrictions---and the final analysis sample. Reported p-values correspond to Pearson chi-squared tests of equality in distributions across groups.

Across most dimensions, observable characteristics are similar between excluded and retained households. While some differences arise along geographic location and household composition, these differences are modest in magnitude, and joint tests fail to reject equality across the majority of demographic characteristics. Overall, the balance results suggest limited scope for selection on observables arising from the final exclusion.

\vspace{-1em}
\subsection{Characteristics summary statistics}

Tables~\ref{tab:nutritional_stats_nutrients} and~\ref{tab:nutritional_stats_binary_attr} report descriptive statistics for the 23 characteristics used in the analysis. Nutritional attributes are expressed per 100g serving. Calories cluster tightly across products, while carbohydrates and sugar are generally high and fat content is low. Sodium exhibits substantial cross-product variation. Among binary attributes, whole-grain indicators and health-related descriptors are common, whereas organic and gluten-free labels are rare. Brand indicators reflect the dominance of Kellogg's and General Mills in the sample.

\begin{table}[!h]
\centering
\caption{\label{tab:nutritional_stats_nutrients}
Nutritional characteristics of cereal products (per 100g)}
\centering
\fontsize{10}{12}\selectfont
\begin{tabular}[t]{>{\raggedright\arraybackslash}p{3cm}>{\raggedright\arraybackslash}p{0.5cm}>{\raggedleft\arraybackslash}p{2cm}>{\raggedleft\arraybackslash}p{2cm}>{\raggedleft\arraybackslash}p{2cm}>{\raggedleft\arraybackslash}p{2cm}>{\raggedleft\arraybackslash}p{2cm}}
\toprule
Characteristic & units & mean & sd & q1 & median & q3\\
\midrule
\cellcolor{gray!10}{calories} & \cellcolor{gray!10}{kcal} & \cellcolor{gray!10}{379} & \cellcolor{gray!10}{38} & \cellcolor{gray!10}{364} & \cellcolor{gray!10}{379} & \cellcolor{gray!10}{400}\\
carbohydrates & g & 81 & 7 & 78 & 82 & 85\\
\cellcolor{gray!10}{total fat} & \cellcolor{gray!10}{g} & \cellcolor{gray!10}{4} & \cellcolor{gray!10}{4} & \cellcolor{gray!10}{2} & \cellcolor{gray!10}{3} & \cellcolor{gray!10}{5}\\
saturated fat & g & 1 & 1 & 0 & 0 & 1\\
\cellcolor{gray!10}{fiber} & \cellcolor{gray!10}{g} & \cellcolor{gray!10}{8} & \cellcolor{gray!10}{5} & \cellcolor{gray!10}{3} & \cellcolor{gray!10}{7} & \cellcolor{gray!10}{10}\\
\addlinespace
protein & g & 8 & 3 & 6 & 7 & 10\\
\cellcolor{gray!10}{sodium} & \cellcolor{gray!10}{mg} & \cellcolor{gray!10}{435} & \cellcolor{gray!10}{239} & \cellcolor{gray!10}{277} & \cellcolor{gray!10}{467} & \cellcolor{gray!10}{600}\\
sugar & g & 25 & 12 & 18 & 25 & 33\\
\bottomrule
\end{tabular}
\footnotesize
\begin{flushleft}
\textit{Notes:} Summary statistics computed across the UPCs in the final sample after excluding two products with missing serving-size entries. Nutritional values are standardised per 100g serving and rounded to the nearest integer.
\end{flushleft}
\end{table}

\vspace{-1em}

\begin{table}[!h]
\centering
\caption{\label{tab:nutritional_stats_binary_attr}
Binary product characteristics}
\centering
\fontsize{10}{12}\selectfont
\begin{tabular}[t]{>{\raggedright\arraybackslash}p{3.5cm}>{\raggedleft\arraybackslash}p{3cm}>{\raggedright\arraybackslash}p{3.5cm}>{\raggedright\arraybackslash}p{3cm}}
\toprule
Characteristic & Share (\%) & Characteristic & Share (\%)\\
\midrule
\cellcolor{gray!10}{whole grain} & \cellcolor{gray!10}{57} & \cellcolor{gray!10}{honey flavour} & \cellcolor{gray!10}{15}\\
organic & 8 & gluten-free & 3\\
\cellcolor{gray!10}{oat-based} & \cellcolor{gray!10}{24} & \cellcolor{gray!10}{Kellogg's} & \cellcolor{gray!10}{18}\\
granola & 10 & General Mills & 13\\
\cellcolor{gray!10}{health halos} & \cellcolor{gray!10}{33} & \cellcolor{gray!10}{Post} & \cellcolor{gray!10}{11}\\
\addlinespace
fruity & 6 & Quaker & 7\\
\cellcolor{gray!10}{nutty} & \cellcolor{gray!10}{11} & \cellcolor{gray!10}{Kashi} & \cellcolor{gray!10}{4}\\
chocolatey & 5 &  & \\
\bottomrule
\end{tabular}
\footnotesize
\begin{flushleft}
\textit{Notes:} Shares denote the percentage of the $K=801$ UPCs in the final sample exhibiting the indicated characteristic. Values are rounded to the nearest integer.
\end{flushleft}
\end{table}

\vspace{-1em}
\subsection{Purchase intensity, brand concentration, and prices}

Figure~\ref{fig:distribution_purchases_per_period_app} shows the distribution of purchase counts per household-period. Figure~\ref{fig:units_exp_density_app} reports the distributions of average period consumption and expenditure. Figure~\ref{fig:brand_hist_app} reports the distribution of the number of distinct brands purchased per period. Figure~\ref{fig:avg_price_density_app} plots the distribution of average unit values across UPCs.

\vspace{-1em}

\begin{figure}[H]
\centering
\vspace{1em}
\caption{Purchase intensity, brand concentration, and prices}
\vspace{-0.5em}
\hspace*{-0.03\textwidth}%
\begin{subfigure}[t]{0.49\textwidth}
\centering
\includegraphics[width=\textwidth]{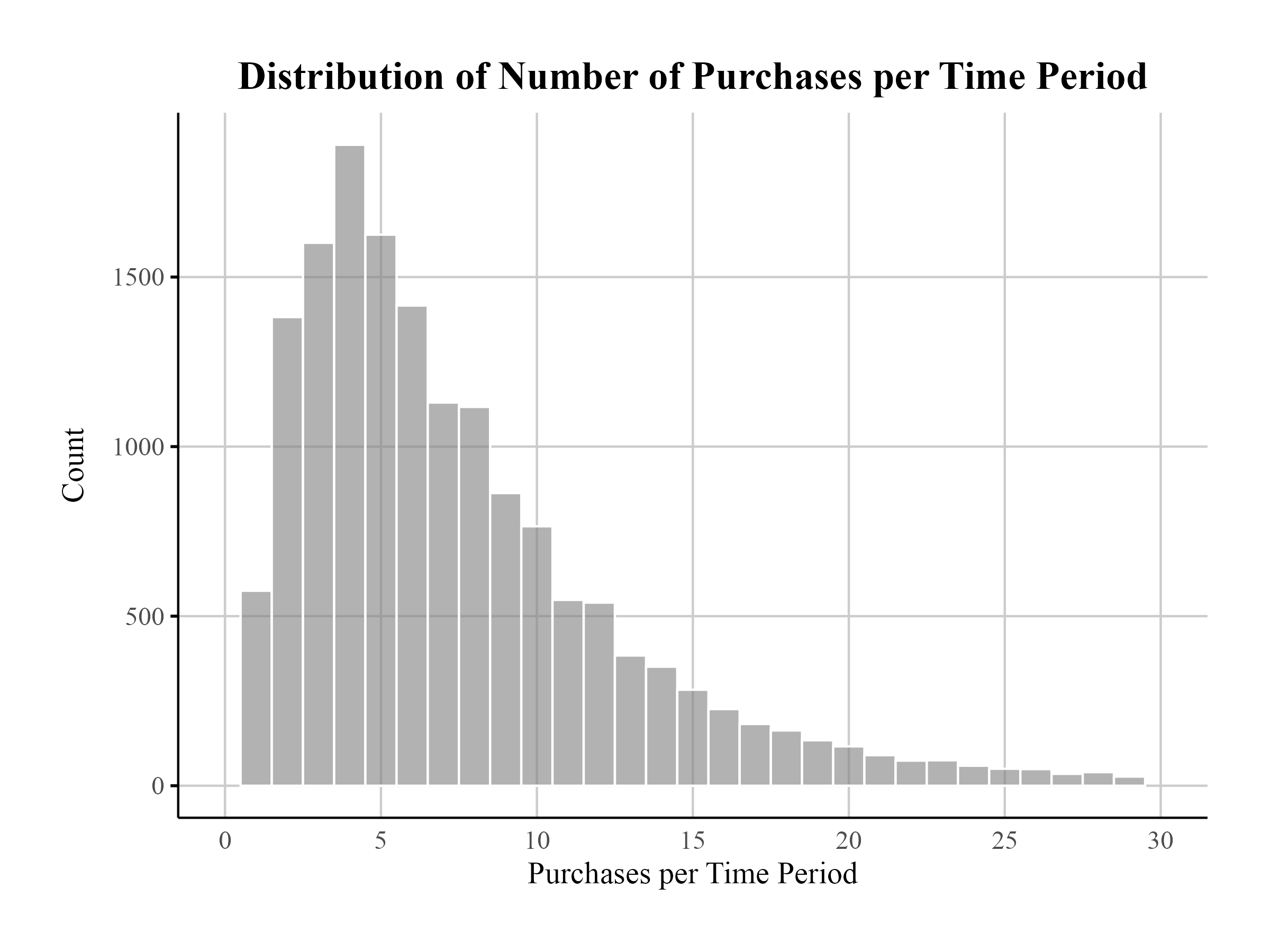}
\caption{Units purchased per household-period}
\label{fig:distribution_purchases_per_period_app}
\end{subfigure}
\hfill
\begin{subfigure}[t]{0.49\textwidth}
\centering
\includegraphics[width=1.1\textwidth]{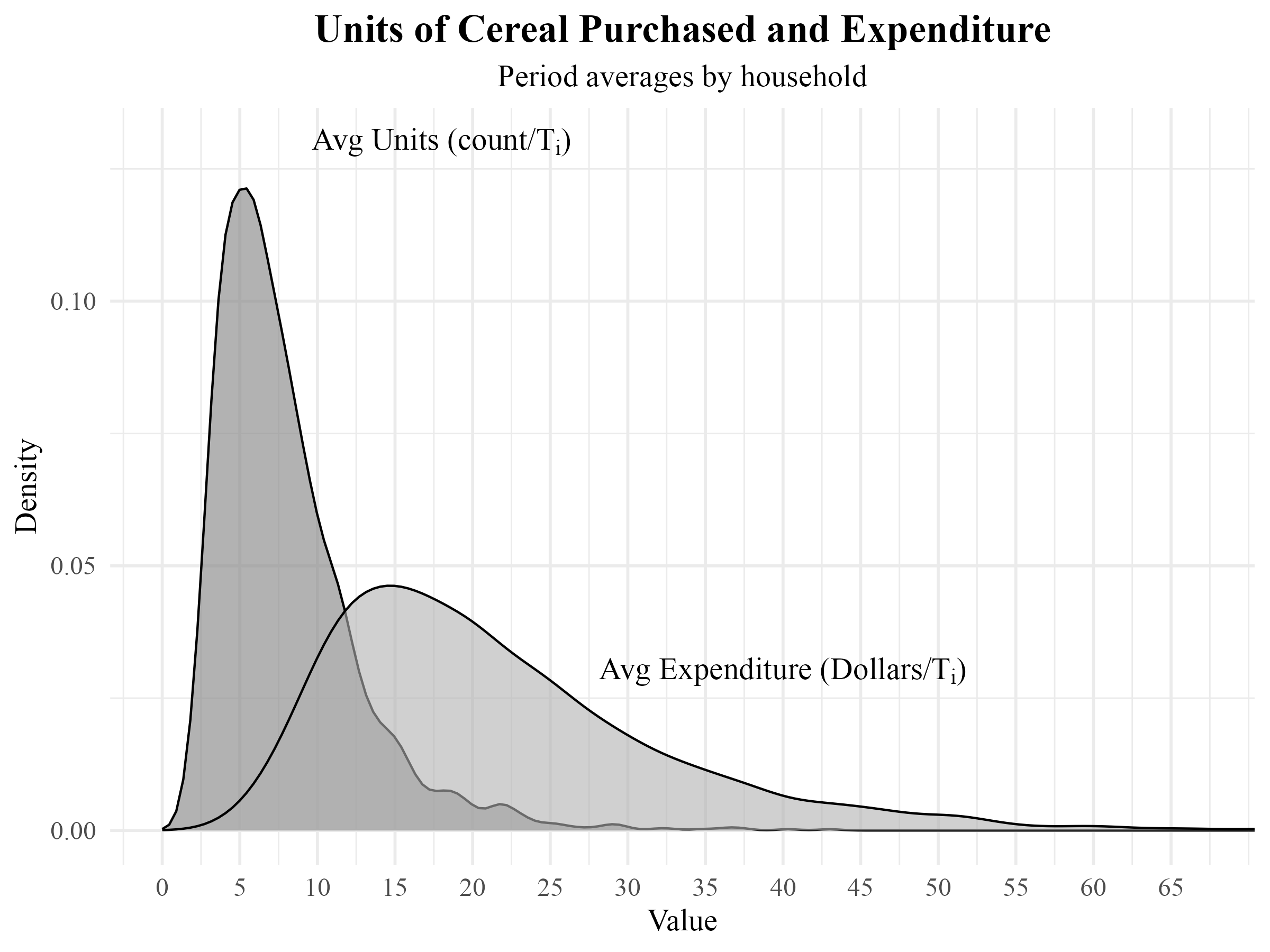}
\caption{Average purchasing intensity across households}
\label{fig:units_exp_density_app}
\end{subfigure}

\vspace{-0.25em}
\hspace*{-0.03\textwidth}%
\begin{subfigure}[t]{0.49\textwidth}
\centering
\includegraphics[width=\textwidth]{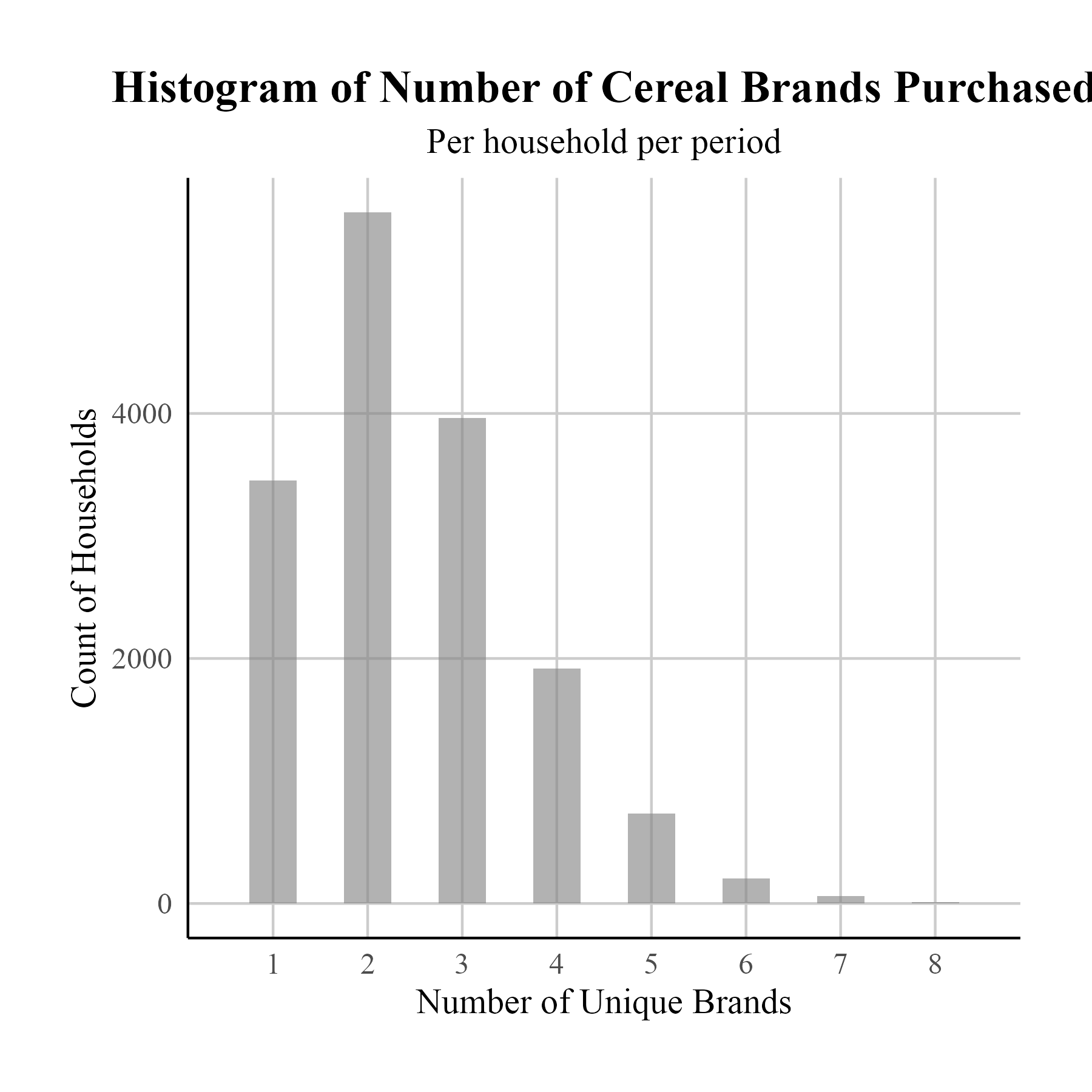}
\caption{Distinct brands purchased per household-period}
\label{fig:brand_hist_app}
\end{subfigure}
\hfill
\begin{subfigure}[t]{0.49\textwidth}
\centering
\includegraphics[width=1.2\textwidth]{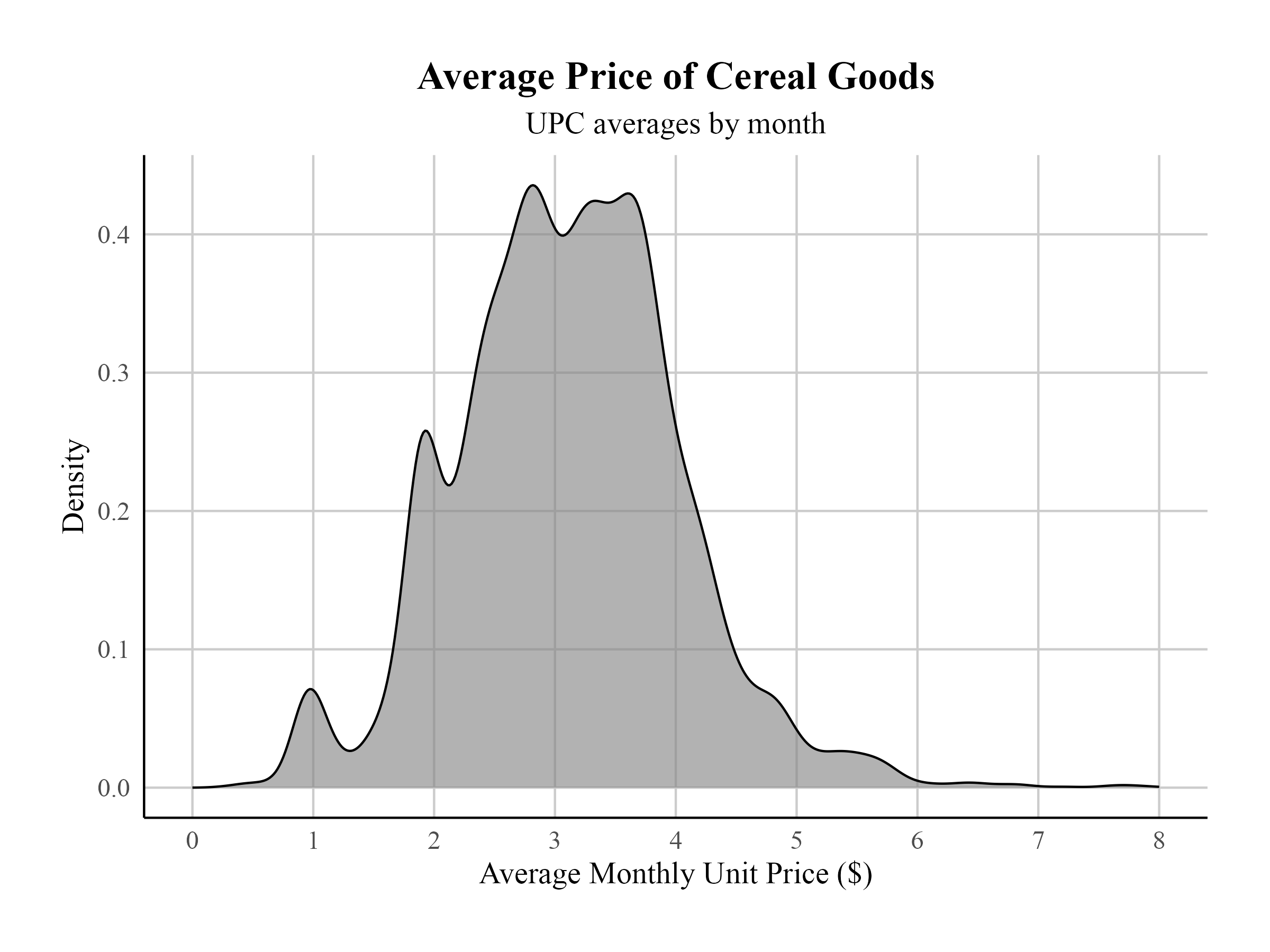}
\caption{Average unit values by UPC}
\label{fig:avg_price_density_app}
\end{subfigure}
\footnotesize
\begin{flushleft}
\textit{Notes:} Panel (a) reports a histogram of total cereal units purchased by household $i$ in period $t$, pooled across all households and constructed time periods. Panel (b) shows the distributions of household-level averages of units purchased per period and expenditure per period across the sample. Panel (c) reports a histogram of the number of distinct cereal brands purchased by a household within a constructed time period. Panel (d) shows the distribution of average unit values (expenditure divided by quantity) computed at the UPC level over the sample period.
\end{flushleft}
\end{figure}

\clearpage
\end{appendices}

\newpage
\bibliographystyle{apalike}
\bibliography{./tex/references}

\end{document}